\documentclass[sigplan,screen]{acmart}

\usepackage{style-files/preamble}

\acmYear{2026}\copyrightyear{2026}
\setcopyright{cc}
\setcctype[4.0]{by}
\acmConference[SOSP '26]{Symposium on Operating Systems Principles}{September 29--October 2, 2026}{Prague, Czech Republic}
\acmBooktitle{Symposium on Operating Systems Principles (SOSP '26), September 29--October 2, 2026, Prague, Czech Republic}
\acmDOI{10.1145/3830418.3843865}
\acmISBN{979-8-4007-2585-2/26/09}

\begin{document}
\title{\sys{}: Safe Interactive Transactions in the Presence of Byzantine Clients}

\author{Austin T. Li}
\affiliation{
    \institution{Cornell University}
    \country{}
}

\author{Daniel H. Lee}
\affiliation{
    \institution{Cornell University}
    \country{}
}

\author{Lorenzo Alvisi}
\affiliation{
    \institution{Cornell University}
    \country{}
}

\author{Natacha Crooks}
\affiliation{
    \institution{UC Berkeley}
    \country{}
}

\author{Florian Suri-Payer}
\affiliation{
    \institution{Cornell University / Databricks}
    \country{}
}

\renewcommand{\shortauthors}{Li et al.}

\begin{abstract}
Byzantine fault-tolerant (BFT) systems are, in principle, an appealing foundation for transactional applications involving mutually distrustful participants.
Yet their adoption has been hampered by two persistent stumbling blocks---performance and developer convenience---which are often in tension with one another.
Recent systems show promising progress on both fronts by shifting to a client-centric architecture; clients execute transactions locally and concurrently, while the system resolves any data conflicts to maintain database serializability.
We argue that, in its current form, this approach introduces a critical vulnerability: it leaves the integrity of the database exposed to Byzantine clients, which may issue malicious or incorrect transactions.
We address this threat with \sys{}, a framework that prevents Byzantine clients from compromising database integrity by executing rogue transactions.
\sys{} combines redundant execution to validate transaction outcomes with a flexible, heterogeneous policy framework for expressing an application's data integrity requirements.
We apply \sys{} to harden several existing BFT database systems and find that it imposes only modest overheads---3\%--16\% in throughput and 3\%--21\% in latency.
\end{abstract}



\maketitle


\section{Introduction}
\label{sec:introduction}

Decentralized applications that support data sharing among mutually distrustful parties are increasingly deployed in domains such as healthcare and financial services~\cite{healthcare-blockchain,healthcare-blockchain2,diem,digital-euro,fnality}.
In these settings, no single participant can be trusted to manage shared state, yet all require strong consistency and integrity guarantees.
Byzantine fault-tolerant (BFT) systems address this tension by replicating state across multiple parties and ensuring agreement on a totally ordered, tamper-evident log of operations~\cite{pbft,autobahn,zyzzyva}.
As a result, BFT systems present an appealing foundation for transactional applications spanning organizational or administrative boundaries.

Despite this promise, BFT systems have seen limited adoption for general-purpose transactional workloads.
Two obstacles have consistently hindered their use: performance and developer convenience.

Traditional BFT designs follow the state machine replication (SMR) model, colocating application logic and state management at backend replicas~\cite{smr}.
This architecture enables strong safety guarantees---replicas deterministically execute application code and directly enforce application semantics.
Over the years, extensive research has improved the performance~\cite{pbft,zyzzyva,hotstuff}, robustness~\cite{aardvark}, and network adaptability~\cite{autobahn} of SMR-based BFT protocols.
However, this approach imposes significant constraints on application developers, who must embed logic into replica-side execution environments, often through custom domain specific languages, and forfeit the ability to scale or evolve application logic independently of the storage layer~\cite{pavlo-sigmod-keynote}.

To reconcile performance and usability, several recent systems have shifted to a client-centric architecture, in which clients (typically front-end servers) execute transactions locally and interactively issue database operations as part of transaction execution~\cite{basil,pesto,hrdb}.
Combined with an appropriate concurrency control mechanism, this design enables high throughput, supports modern transactional abstractions, and significantly improves developer ergonomics. 
The BFT database backend enforces serializability over a stream of client-issued operations, while application logic remains entirely client-side.

We argue that, while it improves performance and usability, this architectural shift unwittingly creates a {\em fundamental} vulnerability to database integrity.

By moving execution out of the trusted replication boundary, client-centric BFT databases implicitly trust clients to execute application logic faithfully.
A Byzantine client, however, may submit transactions that violate application-level invariants, embed malicious logic, or exploit subtle semantic assumptions—without violating serializability.
The BFT database ensures agreement and ordering of operations as issued, but it cannot ensure that those operations correspond to a valid execution of the application.
Thus, client-centric designs abandon a core safety guarantee traditionally provided by SMR-based systems: that all committed transactions reflect correct application execution.

This vulnerability is not an implementation artifact but a structural consequence of today's client-centric architectures.
Figure~\ref{fig:boundary} illustrates the execution boundary in these architectures: the BFT-hardened database enforces only database isolation properties, while application semantics are defined and executed entirely at the client.
As a result, database integrity is fully exposed to arbitrary  updates by Byzantine clients.
Addressing this safety gap requires more than stricter concurrency control---it demands a way to validate application execution itself.

This paper introduces \sys{},\footnote{\textbf{S}afe \textbf{in}teractive \textbf{tr}ansactions. Also, a nod to the process of {\em sintering}, which uses heat to strengthen a material without fully melting it.} a framework that restores to client-centric BFT databases a fundamental guarantee: objects can be updated only by correctly executed application logic.
Rather than trusting the execution of a single client, \sys{} enforces proactive validation: a transaction is accepted only if multiple independent clients re-execute it and agree on its outcome.
These {\em endorsements} serve as proof that the transaction's effects are consistent with correct application execution.

At the same time, na\"{i}vely validating all transactions redundantly would impose prohibitive overhead.
\sys{} therefore introduces a heterogeneous, object-level policy framework that allows applications to precisely specify their integrity requirements.
Different data objects may demand different levels of scrutiny---for example, financial records may require stronger validation than auxiliary metadata.
Policies determine how many validation clients must endorse transactions that modify a given object, ensuring that overhead scales with required safety rather than uniformly across the workload.

Designing such a framework raises several challenges. First, \sys{} must minimize redundant database access during re-execution to avoid overwhelming the backend.
It addresses this by forwarding sufficient execution context to validators, while ensuring Byzantine clients cannot manipulate this information to forge endorsements.
Second, transactions frequently span multiple objects with different policies; \sys{} must ensure that endorsements jointly satisfy all relevant object policies and, further, prevent objects with weaker integrity requirements from influencing objects with stronger ones.
Finally, integrity requirements may evolve over time.
\sys{} supports dynamic policy updates through {\em governance transactions} that take effect atomically during normal execution, treating policies themselves as first-class data objects.

We have implemented \sys{} by hardening several existing BFT systems, including Basil~\cite{basil}, Pesto~\cite{pesto}, and HotStuff~\cite{hotstuff} and \bftsmart{}~\cite{bftsmart} layered over transactional data stores with both SQL and key-value interfaces.
Across standard transactional workloads (\tpcc{}, \seats{}, and \smallbank{}), \sys{} incurs modest overheads of 3\%--16\% in throughput and 3\%--21\% in latency.
Additional microbenchmarks characterize how validation policies affect performance and quantify the residual influence Byzantine clients retain under different configurations.
In summary, this paper makes the following contributions:

\begin{itemize}
  \item We identify and formalize {\em application execution validity}, a safety guarantee lost by existing client-centric BFT database designs.  
  \item We present \sys{}, a validation-based framework that restores this guarantee using heterogeneous, object-level integrity policies. 
  \item We demonstrate that \sys{} can be retrofitted onto multiple BFT database systems with low overhead, significantly strengthening their safety against Byzantine clients.
\end{itemize}

\begin{figure}[t]
    \centering
    \includegraphics[width=\columnwidth]{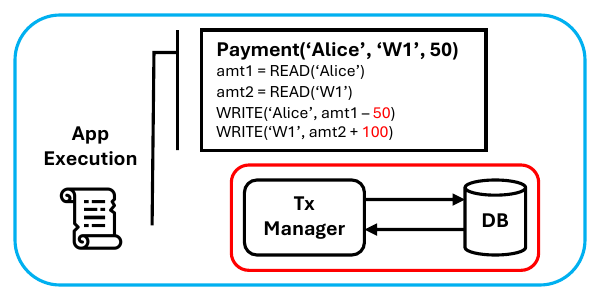}
    \caption{
        The safety guarantees of today's client-centric BFT databases only extend to the isolation properties of the transactions being executed, not to their semantics: nothing prevents Byzantine clients from issuing incorrect transactions ({\em e.g.}, \texttt{Payment} transactions that incorrectly update the customer and warehouse balances).
    }    

    \label{fig:boundary}
\end{figure}



\section{The Pitfall}
\label{sec:pitfall}

State machine replication (SMR) has long served as the foundation for building fault-tolerant distributed systems. 
In SMR, a system is modeled in terms of two components:  \emph{state variables}, which capture the system's state, and  \emph{commands}, which deterministically transform that state. 
Traditional SMR designs assign responsibility for both command execution and state maintenance to a set of replicas and require correct replicas to execute the same sequence of commands in the same order\footnote{
    The total-order requirement can be relaxed when commands commute~\cite{smr}. For simplicity, and without loss of generality, we ignore this optimization.} and update their local state accordingly. 

When applied to transactional data stores, this paradigm requires database replicas to both store shared state ({\em i.e.}, the values of database objects) and execute transactions. 
Concretely, correct replicas execute transactions whose reads and writes \one{} arise from valid, application-defined transaction logic and \two{} collectively satisfy the required isolation guarantees ({\em e.g.}, serializability).
In Byzantine settings, replicas further vote on proposed updates so that only sufficiently endorsed state transitions are applied to the database.

Architectures that embed application logic directly within the database—most notably via stored procedures—naturally satisfy these requirements. 
Because replicas execute application logic themselves, they can directly validate that all state updates arise from correct executions of approved transactions. 
However, this tight coupling comes at a steep cost: stored procedures are typically written in vendor-specific languages with limited tooling, making application logic harder to develop, test, maintain, and evolve. 

Client-centric BFT architectures were developed to overcome these limitations—and in many respects, they succeed. 
They allow application logic to be written in general-purpose languages and have been shown to achieve performance competitive with unreplicated databases on standard transactional workloads~\cite{basil,pesto}. 
Achieving these benefits, however, requires a subtle but consequential shift in how the SMR abstraction is instantiated. 

In these architectures, correct database replicas remain responsible for storing state and apply agreed-upon updates, but they no longer execute the commands that generate those updates. 
Instead, clients execute transactions, interleaving application logic with database operations, and submit only the resulting reads and writes to the replicated backend. 
As a result, replicas can enforce consistency properties over the operations they observe, but they can no longer determine whether those operations arise from correct executions of application-defined transactions. 
This shift breaks a foundational stipulation of the SMR contract. 
Once application logic resides outside the trusted replication boundary, replicas lose the ability to ensure that state variables are updated only by valid  executions of state machine commands.

Although we frame this discussion in terms of SMR, the issue we identify is not specific to state machine replication.

It arises more generally in BFT systems that decouple application execution from replicated agreement. 
In client-driven Byzantine storage and quorum-based systems~\cite{byzantine-quorum-systems,q/u,HQ,smallQS}, replicas ensure consistency over stored values or updates, but likewise lack visibility into the semantic correctness of client-side computation. 
When such abstractions are used as substrates for transactional or application-level systems, they inherit the same vulnerability: replicas may agree on updates without being able to validate that those updates result from correct executions of application logic. 
SMR makes the loss explicit by contrast with its stronger guarantees, but the underlying problem is architectural rather than SMR-specific.

Correctness notions rooted solely in ordering and isolation of reads and writes are  insufficient in client-centric architectures, as they fail to capture whether application logic itself was executed correctly. 
Addressing this gap requires an explicit model of application execution and a corresponding safety definition that extends beyond traditional database consistency.

\section{Model and Definitions}
\label{sec:model}

\subsection{System Model}
\label{subsec:model-sys}

\sys{} inherits the standard assumptions of prior BFT work~\cite{pbft,zyzzyva,basil,pesto}.
\sys{} assumes partial synchrony only to ensure liveness~\cite{partial-sync}.
A participant is \emph{correct} if it adheres to the protocol specification, and \emph{faulty} otherwise; faulty participants may deviate arbitrarily from their specification.
The adversary is strong but static: it may coordinate the actions of faulty clients and servers subject to the underlying fault thresholds, but cannot break standard cryptographic primitives.
We assume that all clients and servers are explicitly registered with the system and can be authenticated via signatures. We denote a signed message $m$ as $\langle m \rangle_\sigma$. 

\subsection{Safety with Interactive Transactions}

Traditional correctness definitions for databases are intentionally limited in scope: they characterize safety exclusively in terms of reads and writes, typically through isolation guarantees such as serializability~\cite{concurrency-control-recovery-db,critique-ansi-sql-isolation}.
These definitions align naturally with classical SMR systems, where replicas execute application logic themselves and thus implicitly validate command execution.

However, when application logic executes at clients, outside the replicated core, a Byzantine client may issue a sequence of reads and writes that is serializable yet does not correspond to any valid execution of an application-defined transaction.
For example, the \texttt{Payment} transaction in Figure~\ref{fig:boundary} remains serializable even when it violates fundamental application invariants.
Reasoning about such behaviors requires correctness definitions that explicitly account for application execution, not merely for database-level consistency.


\myhdr{Application Model}
Most applications predefine a set of transactions as part of their application logic.
Developers write this logic in a general-purpose programming language, where transactions interleave database requests with client-side computation.
We adopt this development model and refer to such transactions as \emph{interactive transactions}.
This approach is not only developer-friendly,  but allows applications to scale independently of the database backend~\cite{pavlo-sigmod-keynote}.

Formally, we model an application with interactive transactions as follows.


\begin{definition}\label{def:application}
    \defApplication{}
\end{definition}

The set of transactions $\mathcal{T}$ constrains which executions are correct for the application.
Systems that colocate application execution with application state can guarantee that only predefined transactions commit, and that each such transaction is executed in accordance with its specification~\cite{smr}.
In contrast, existing BFT systems with interactive transactions implicitly trust authenticated clients (some of which may be Byzantine!) to execute application logic correctly~\cite{basil,pesto}.
As argued in Section~\ref{sec:pitfall}, this trust is misplaced.

We therefore formalize safety directly in terms of application execution.

\begin{definition}\label{def:validity}
    \defValidity{}
\end{definition}

Invalid read-only transactions affect only the issuing client and thus we intentionally make no assumptions on them.
Application execution validity captures the core safety guarantee that is lost when shifting from server-side execution to client-side interactive transactions. 
\sys{} focuses specifically on restoring this core guarantee, enabling applications to adopt client-side interactive transactions without sacrificing this fundamental safety guarantee. 
Other concerns---such as ordering transactions execution, enforcing isolation levels, or validating a transaction's input---are orthogonal and addressed by existing mechanisms.

\section{\sys{} Framework}

This section presents the \sys{} framework, which restores \emph{application execution validity} in client-centric BFT databases. 
Building on the model introduced in Section~\ref{sec:model}, \sys{} provides fault tolerance against Byzantine clients by validating transaction execution at the application level, rather than relying solely on database-level consistency. 
We first introduce a transaction-centric fault model, then describe the safety and liveness properties \sys{} guarantees, and finally give a high-level overview of its operation.


\subsection{A Transaction-Centric Fault Model}
\label{subsec:framework-fault-model}

Defending against Byzantine clients requires introducing redundancy into transaction execution. 
A natural approach is to adopt a client-centric fault model analogous to replica-based fault tolerance~\cite{pbft,zyzzyva}, through heterogeneous trust assumptions~\cite{heterogeneous-paxos} or  by assuming a bound on the number of faulty clients system-wide. \sys{} in particular leverages the fact that its clients are typically front-end servers controlled by the same organization(s) that run its applications to assume the existence of a bound on the number of faulty clients system-wide. 
Such fault models, however, are coarse-grained: they treat all transactions and all application data uniformly, regardless of their semantics or importance. 
\sys{} instead adopts a \emph{transaction-centric} fault model that expresses fault tolerance in terms of \emph{objects} and the \emph{actions} transactions perform on them. 
This model allows applications to pay for safety selectively, enforcing stronger guarantees only where application semantics demand them. 
Our design draws inspiration from integrity labels in the information-flow control literature~\cite{biba-integrity} and instantiates them concretely as {\em integrity policies}. 

\myhdr{Integrity Policies}
Each object is assigned an integrity policy that reflects the maximum degree of adversarial collusion the application expects that object to face. Intuitively, an object with policy $x$ is protected against up to $x$ colluding Byzantine clients. 

Importantly, integrity policies need not be uniform. 
Applications commonly contain data with different security and correctness requirements. 
For instance, in the \tpcc{} benchmark~\cite{tpcc}, financial records ({\em e.g.}, customer balances) are significantly more sensitive than record-keeping data such as inventory counts, which can be reconciled through external validation. 
\sys{} allows applications to reflect these distinctions directly in their safety policy. 



\myhdr{Enforcing Policies}
\sys{} enforces integrity policies through \emph{transaction endorsement}.
To modify an object governed by policy $x$, a client must collect at least $x$ other independent client endorsements that corroborate its transaction execution. 
This ensures that coalitions of $x$ or fewer Byzantine clients cannot corrupt the object. 

Allowing heterogeneous policies introduces a new challenge: policies cannot be assigned independently of application semantics. 
If a transaction reads a low-policy object and uses it to influence a write to a high-policy object, then the effective integrity of the high-policy object is weakened. 
To prevent this, \sys{} adopts the principle of \emph{noninterference} from information-flow control~\cite{noninterference,language-based-ifc}.

\begin{definition}\label{def:noninterference}
    \defNoninterference{}
\end{definition} 

Transactions that violate noninterference are rejected by default.

\myhdr{Relaxing Noninterference} 
Strict noninterference is often too restrictive in practice~\cite{noninterference,language-based-ifc}. 
Some transactions legitimately transform low-integrity data into high-integrity updates by validating it against trusted information. 
In \tpcc{}, for example, a \texttt{Delivery} transaction may read order records (low policy) and update customer balances (high policy) after performing application-specific validation (Figure~\ref{fig:app-policy}).

\sys{} supports such patterns by allowing transactions to \emph{lift} the policy of selected read objects. 
Lifting is an explicit, application-controlled operation: objects that are lifted are treated as having a higher policy for information flow purposes during the transaction. 
For simplicity, \sys{} considers lifted objects as having policy as high as any object the transaction writes to, effectively excluding them from noninterference for that transaction.
Applications express lifting logic using deterministic \emph{lift functions}, which validate low-policy data according to the application's semantics. 
Algorithm~\ref{alg:lift-func} illustrates a lift function for validating orders against payments. Lift functions offer a clean and fully general mechanism to express semantic checks that, in their absence, one would attempt to encode within a transaction or enforce through  database consistency constrains.


Lifting is optional and conservative by default: transactions that neither satisfy noninterference nor explicitly lift required objects are rejected. 
Because lifting weakens policy protections, it is intended to be used sparingly and with care.

\begin{figure}[t]
\centering 
\includegraphics[width=\columnwidth]{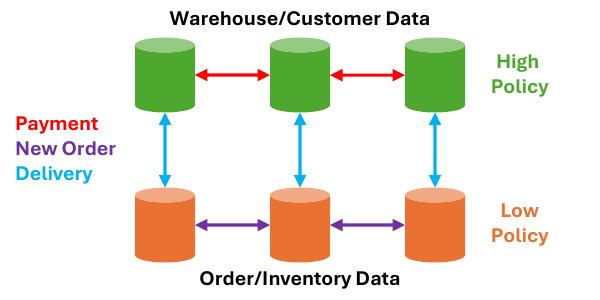} \caption{Example object policies and transactions for a transactional application ({\em e.g.}, \tpcc{}).} \label{fig:app-policy}
\end{figure} 

\begin{algorithm}[t]
\caption{
    Example Lift Function for \texttt{Delivery} $(T)$
}\label{alg:lift-func}
\begin{algorithmic}[1]
\State $\textit{lift} = []$
\For{$\textit{order} \in T$}
    \If{$\exists \ \textit{payment} \in T : \textit{payment.order} == \textit{order}$}
        \State $\textit{lift}.append(\textit{order})$
    \EndIf
\EndFor
\State \Return $\textit{lift}$
\end{algorithmic}
\end{algorithm}

\subsection{\sys{} Properties}
\label{subsec:framework-properties}

Under the transaction-centric fault model described above, \sys{} guarantees \emph{application execution validity}. 

\begin{theorem}\label{thm:sys-safety}
    \thmSysSafety{}
\end{theorem}

This guarantee ensures that, given a sufficiently strong policy, all committed, non-read-only transactions correspond to faithful executions of application-defined transaction logic, even in the presence of Byzantine clients. 
\sys{} does not, however, prevent Byzantine actors from obstructing progress. 
Liveness is guaranteed only when correct clients can satisfy required policies by collecting endorsements from correct clients (without relying on responses from faulty ones).
We defer detailed proofs of correctness to an extended technical report~\cite{sintr-tr}.


\subsection{\sys{} Overview}
\label{subsec:framework-overview}

\begin{figure*}[t]
    \centering
    \includegraphics[width=0.8\textwidth]{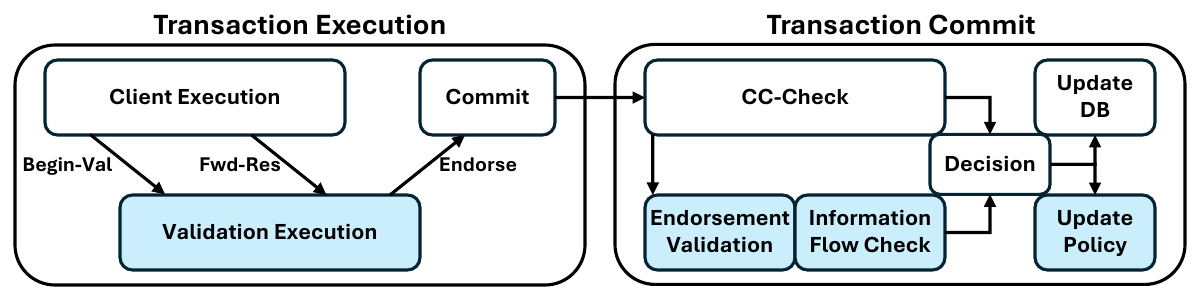}
    \caption{Overview of transaction execution and validation in \sys{}.}
    \label{fig:overview}
\end{figure*}

Figure~\ref{fig:overview} summarizes transaction processing in \sys{}. 
Clients execute transactions locally while requesting parallel validation from peer clients. 
To avoid additional database load, validation clients do not issue database requests and instead rely on the initiating client to provide the necessary execution context and responses, along with proofs certifying their correctness.

Upon completing their re-execution, validation clients return signed endorsements to signify that the initiating client performed the transaction candidly. 
The initiating client shares these endorsements with database replicas as part of its transaction commit request. 
Replicas then perform standard concurrency-control checks while additionally validating endorsements and enforcing information-flow constraints. 
Only transactions passing all checks are committed.
\sys{} also supports \emph{governance transactions} that update integrity policies during normal execution.

\sys{} is designed for environments in which Byzantine clients are expected but limited. 
Clients are authenticated, transactions can be audited~\cite{peerreview,iaccf}, and collecting endorsements forces Byzantine clients to provably collude. 
While \sys{} defaults to optimistic integrity policies, applications may adopt more conservative policies when stronger guarantees are required.

\section{\sys{} Protocol}

\sys{} is a general purpose framework, and can be used to augment any BFT database system supporting interactive transactions.
Our evaluation instantiates \sys{} atop several different BFT systems (Basil~\cite{basil}, Pesto~\cite{pesto}, and SMR-based transactional systems using HotStuff~\cite{hotstuff} and \bftsmart{}~\cite{bftsmart}).
Here, we abstract protocol-specific functionality and assume the underlying system exposes a standard transaction interface.
Each database object is identified by a unique \textit{key} ({\em e.g.}, a row's primary key in a relational database), and each \textit{key} is associated with a concrete policy.
For simplicity, policies are initially fixed across a key's lifetime; we discuss in Section~\ref{subsec:design-gov-txn} how to adjust policies at runtime.

\subsection{Client Endorsement Collection Protocol}
\label{subsec:design-endorse-collect}

Clients execute transactions by interleaving database requests with application logic. 
To enforce application-level policies in the presence of Byzantine clients, \sys{} requires transaction writes to be {\em endorsed} by other clients that independently validate the transaction's execution. 

Endorsements provide cryptographic evidence that a transaction's externally visible effects are consistent with the application specification: they attest to the correctness of the transaction's writes and of the database reads that influence them. 
Reads that do not affect any writes produce no external effect and require no endorsements. 

The endorsement collection protocol has two goals.
First, a Byzantine initiating client must not be able to fabricate, reuse, or manipulate endorsements to commit an invalid transaction. 
Second, endorsement and validation overheads should be minimal.

\sys{} organizes endorsement collection around the phases of transaction execution: beginning a transaction, issuing database requests, executing application logic, and committing. 
Throughout execution, initiating and validation clients record the transaction's \emph{read set} (keys accessed) and \emph{write set} (keys modified).
These sets conservatively approximate the information flow of the transaction and are later used to validate endorsements and enforce policy constraints.

We use Figure~\ref{fig:endorse-collect-fail} as a running example. 
When the initiating client diverges from correct execution, correct validation clients compute non-matching endorsements, causing replicas to abort the transaction at commit.

\myhdr{\texttt{Begin}{\bf: Initiating client}} When starting a transaction, the initiating client $C$ estimates the set of validation clients $V = \{C_V\}$ from which it must obtain endorsements. 
This estimate is based on the anticipated write set of the transaction and the policies associated with those keys, potentially using application-provided hints. 
Because the write set of an interactive transaction may not be fully known in advance, this estimate is best-effort; we assume for now that $V$ is sufficient and defer how to dynamically discover and adapt $V$ to Section~\ref{subsec:design-practical}.

\protocol{\textbf{C} $\to$ \textbf{C\textsubscript{V}}: 
Client $C$ sends begin validation requests.}
To initiate validation, $C$ sends a request $\beginval{} := \langle S_T,\ \txid{} \rangle$ to each $C_V \in V$, where $S_T$ describes the transaction's initial state ({\em e.g.}, name and input parameters) and $\txid{}$ is a unique identifier used to prevent endorsement reuse.
In Figure~\ref{fig:endorse-collect-fail}, the initiating client seeks validation for a \texttt{Payment} transaction from `Alice' to warehouse `W1' for \$50.

\protocol{\textbf{C\textsubscript{V}} $\leftarrow$ \textbf{C}: Validation clients begin execution.}
\myhdr{\texttt{Begin}{\bf: Validation clients}} Upon receiving a \beginval{} message, each validation client $C_V$ begins local execution of the specified transaction using the provided inputs.
Validation clients execute application logic eagerly but block on database requests until results are forwarded by the initiating client.


\begin{figure*}[t]
    \centering
    \includegraphics[width=\textwidth]{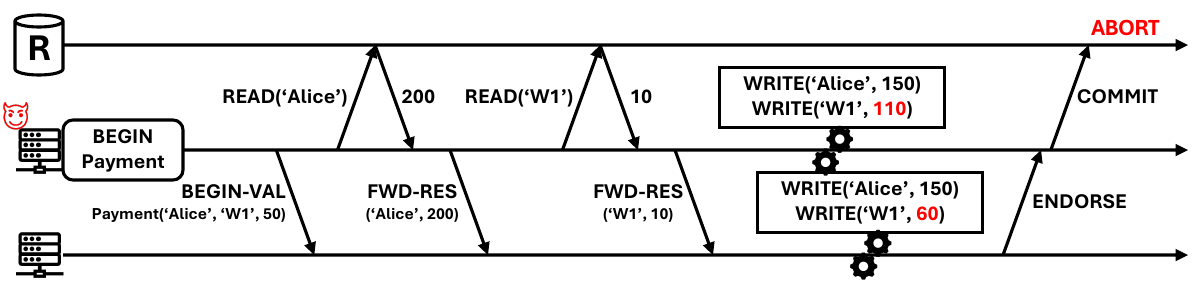}
    \caption{
        Invalid transactions from Byzantine clients fail to gather matching endorsements, and abort at correct replicas.
    }
    \label{fig:endorse-collect-fail}
\end{figure*}

\myhdr{\texttt{Database Request}} Validation clients do not issue database requests to the backend. 
Instead, the initiating client forwards the result of each database request along with a cryptographic proof produced by the underlying BFT system. 
This design serves two purposes. 
First, it avoids increasing load on the backend: redundant work is performed only at the clients, which can be scaled independently. 
Second, it ensures consistent execution across clients, as, without additional server-side support, independent reads issued by validation clients could return different results: executions could diverge even when all parties are correct.


Because validation clients do not trust the initiating client, forwarded results must be verifiable. 
\sys{} leverages the fact that BFT database systems already return replies accompanied by proofs of correctness ({\em e.g.}, a quorum of matching responses with signatures). 
To prevent equivocation, each forwarded result is also bound to a unique request identifier.

\protocol{\textbf{C} $\to$ \textbf{C\textsubscript{V}}: Client $C$ forwards database request results.}
\myhdr{Initiating client}
For each database request issued during transaction execution, the initiating client forwards a result message
$\fwdres{} := \langle \reqid{},\ \reqres{},\ \reqproof{} \rangle$
to all validation clients. 
Here, $\reqid{}$ is a cryptographic hash of the request and transaction identifier, $\reqres{}$ is the database response, and $\reqproof{}$ is the accompanying proof of correctness.
The response also includes the request's read and write sets, allowing validation clients to track how database state influences the transaction.

\protocol{\textbf{C\textsubscript{V}} $\leftarrow$ \textbf{C}: Validation clients receive forwarded results.}
\myhdr{Validation clients}
Upon receiving a forwarded result, a validation client verifies the proof and checks that it certifies both the response and its association with $\reqid{}$. 
If verification succeeds, the client resumes execution using the forwarded result; otherwise, it cannot produce a valid endorsement.

In Figure~\ref{fig:endorse-collect-fail}, the initiating client forwards database results for the customer and warehouse balance reads. The correct validation client verifies the proofs and uses these values during local execution.

\myhdr{\texttt{Execute}} 
Between database requests, the initiating client and all validation clients execute transaction logic locally. 
\sys{} assumes transaction execution is deterministic given identical inputs and database results; as a consequence, correct clients will reach the same intermediate states and final outcome. 
If the initiating client deviates from correct execution---for example, by computing an incorrect write based on valid reads---correct validation clients will independently compute a different transaction state. 

In Figure~\ref{fig:endorse-collect-fail}, the Byzantine initiating client computes an incorrect warehouse balance update, while the correct validation client, executing the same logic over the same reads, obtains a different value---a mismatch that will eventually cause correct replicas to abort the transaction.

Upon completing execution, each client also deterministically evaluates the transaction's \emph{lift function}.
The lift function takes as input the final transaction state (including the read and write sets) and the associated policies, and returns a set of keys to lift.
For each lifted key, its effective policy is raised to the maximum write-set policy for the transaction, excluding it from subsequent information-flow constraints.



\myhdr{\texttt{Commit}} 
After completing execution, the initiating client collects endorsements from its validation clients.

\protocol{\textbf{C\textsubscript{V}} $\to$ \textbf{C}: Validation clients send endorsement.}
\myhdr{Validation clients} 
Once a validation client $C_V$ finishes local execution, it computes a transaction hash $H_T$ over the transaction identifier, read set, and write set.
Any lifted keys are excluded from the read set prior to hashing.
The client then sends an endorsement $\finishval{} := \langle H_T \rangle_{\sigma_{C_V}}$ to the initiating client, signed using its public key. 
This endorsement attests to $C_V$'s view of the transaction's externally visible effects.

\protocol{\textbf{C} $\leftarrow$ \textbf{C\textsubscript{V}}: Client $C$ assembles endorsements.}
\myhdr{Initiating client}
$C$ independently computes $H_T$.
If it forwarded correct database results and executed the transaction correctly, its hash will match those of all correct validation clients. 
A Byzantine initiating client cannot alter the transaction's read set, write set, or lifted keys without causing an endorsement mismatch with correct clients.

Once $C$ has collected a sufficient number of matching, well-formed endorsements to satisfy the policies associated with all keys in the write set, it begins the transaction commit process.
We assume here that all validation clients reply correctly; we defer discussion of faulty validation clients to Section~\ref{subsec:design-practical}.

\protocol{\textbf{C} $\to$ \textbf{R}: Client $C$ begins commit process.}
To commit, $C$ sends an augmented commit message
{\small $\commit{} := \langle T,\ \txid{},\ \readset{T},\ \writeset{T},\ \sysEndorses{T} \rangle$} to the backend replicas. 
It contains all information required for replicas to validate endorsements and enforce application policies. 
In Figure~\ref{fig:endorse-collect-fail}, the validation client computes a hash that differs from the Byzantine initiating client's, which nonetheless attempts to commit the transaction attaching the mismatched endorsement.


\subsection{Transaction Commit}
\label{subsec:design-txn-commit}

\sys{} augments the underlying transaction commit protocol with additional policy-related checks, but otherwise leaves concurrency control (CC) and atomic commit unchanged.
Replicas vote to commit a transaction only if it satisfies both the backend's native CC checks and \sys{}'s policy enforcement.
When a replica receives a commit message
it performs two types of validation:
endorsement validation and information-flow checks.
A replica may vote to commit only if both succeed.

\myhdr{Endorsement Validation} 
A replica first computes the transaction hash $H_T$ from the commit message. 
It then verifies each endorsement's signature and checks that its hash matches $H_T$, collecting all valid endorsements that do so
(Algorithm~\ref{alg:endorse-val}, Lines~\ref{alg:endorse-val-line:for-endorse}--\ref{alg:endorse-val-line:endfor-endorse}).
For each key in the transaction's write set, the replica retrieves its policy and checks that the number of valid endorsements meets the policy requirement (Lines~\ref{alg:endorse-val-line:for-write}--\ref{alg:endorse-val-line:endfor-write}).
If any write's policy is not satisfied, the replica votes to abort.
In Figure~\ref{fig:endorse-collect-fail}, the transaction contains an endorsement whose hash does not match the replica-computed $H_T$, causing endorsement validation to fail and the transaction to abort.


\begin{algorithm}[t]
\caption{Endorsement Validation ($T$)}\label{alg:endorse-val}
\begin{algorithmic}[1]
\State $E_T = \{\}$
\For{$e \in \sysEndorses{T}$}\label{alg:endorse-val-line:for-endorse}
    \If{$\mathit{e.sig} \text{ is valid} \land \mathit{e.hash} == H_T$}
        \State $E_T = E_T \cup e$
    \EndIf
\EndFor\label{alg:endorse-val-line:endfor-endorse}
\For{$\mathit{key} \in \writeset{T}$}\label{alg:endorse-val-line:for-write}
    \State $p = \mathit{GetPolicy(key)}$
    \If{$|E_T| < p$}
        \State \Return $\False$
    \EndIf
\EndFor\label{alg:endorse-val-line:endfor-write}
\State \Return $\True$
\end{algorithmic}
\end{algorithm}

\begin{algorithm}[t]
\caption{Policy Information Flow ($T$)}\label{alg:policy-leak}
\begin{algorithmic}[1]
\State $P_T = \max(\{\mathit{GetPolicy(key)}: \ \mathit{key} \in \writeset{T}\})$
\For{$\mathit{key} \in \readset{T}$}
    \State $p = \mathit{GetPolicy(key)}$\label{alg:policy-leak-line:get-read-policy}
    \If{$p < P_T$}\label{alg:policy-leak-line:if-implies}
        \State \Return $\False$
    \EndIf\label{alg:policy-leak-line:endif-implies}
\EndFor
\State \Return $\True$
\end{algorithmic}
\end{algorithm}

\myhdr{Policy Information Flow}
Endorsement validation ensures that transaction writes are explicitly approved by a sufficient number of clients, but it does not by itself prevent a transaction from indirectly circumventing policies. 
In particular, when objects have heterogeneous policies, a transaction could use reads protected by weaker policies to influence writes protected by stronger ones, effectively committing a high-policy write with insufficient endorsement.
To prevent such policy violations, \sys{} enforces a conservative information-flow check at replicas. 
Rather than requiring clients to track fine-grained data dependencies during execution, replicas approximate information flow using the transaction's read and write sets. 
For a transaction $T$, let $P_T$ denote the maximum policy among all keys in its write set.
The replica verifies that every key in the read set has policy at least $P_T$ (Algorithm~\ref{alg:policy-leak}). 
Keys explicitly lifted by the transaction's lift function are excluded from the read set for this check, as lifting conceptually raises their effective policy to $P_T$. 
If any non-lifted read violates this constraint, the replica votes to abort the transaction.

\subsection{Governance Transactions}
\label{subsec:design-gov-txn}
\sys{} supports dynamic updates to application policies at runtime.
Doing so requires careful coordination between policy evolution and transaction validation to ensure that transactions are evaluated against the correct policy versions and that policy changes cannot be bypassed. 
To this end, \sys{} represents policies as versioned data maintained by the database and updated through transactions. 
We refer to transactions that modify policy values as {\em governance transactions}.

\myhdr{Versioned Policy Store} 
\sys{} introduces a layer of indirection between database keys and concrete policies. 
Each database key is associated with a {\em policy id}, and each policy id corresponds to a versioned policy value. 
Each replica maintains a {\em versioned policy store}, implemented as ordinary database state ({\em e.g.}, a dedicated table). 
Policy ids may be assigned at arbitrary granularity: a single key, a table, or any application-defined grouping may map to the same policy id. 
This indirection allows policies to evolve independently of application data.
Conceptually, both the policy values themselves and the mapping can be updated.
Our prototype implementation, however, currently assumes fixed mappings, so we discuss only updates to policy values for simplicity. 

\myhdr{Policy Versions and Validation}
Versioning affects both the view clients may have of policies and the checks performed by replicas. 
We defer discussion of stale client views to Section~\ref{subsec:design-practical} and focus here on replica-side validation.
Replica checks use policy versions consistent with the transaction's execution. 
For a transaction $T$ with timestamp $ts_T$, writes are evaluated against the latest policy version with timestamp earlier than $ts_T$. 
For information-flow checks, each read is evaluated against the policy version that was current when the read value was produced. 

\myhdr{Transaction Processing and Isolation} 
Governance transactions must be committed by all correct replicas and respect database isolation. 
In particular, \sys{} requires governance transactions to be serializable with respect to regular transactions, ensuring that committed writes cannot be invalidated retroactively by later policy changes. 
Conceptually, governance transactions are ordinary transactions over policy objects. 
The key additional mechanism is that regular transactions {\em implicitly read} the policy ids associated with their writes.  
These implicit reads capture the transaction's dependency on specific policy versions and allow standard concurrency control to detect conflicts with governance transactions. 

Algorithm~\ref{alg:gov-txn-cc} illustrates an instantiation using optimistic concurrency control (OCC). 
For a regular transaction $T$, the replica records an implicit read of the policy id (and version) associated with each key in $T$'s write set
(Lines~\ref{alg:gov-txn-cc-line:for-add-implicit}--\ref{alg:gov-txn-cc-line:endfor-add-implicit}).
During validation, $T$ aborts if a governance transaction committed a newer version of any such policy between $T$'s read and commit timestamps
(Lines~\ref{alg:gov-txn-cc-line:for-regular-vs-gov}--\ref{alg:gov-txn-cc-line:endfor-regular-vs-gov}).

Conversely, a governance transaction $G$ aborts if it updates a policy id that was implicitly read by a concurrently committed regular transaction
(Lines~\ref{alg:gov-txn-cc-line:for-check-implicit}--%
\ref{alg:gov-txn-cc-line:endfor-check-implicit}).
This bidirectional conflict detection ensures serializability between policy updates and policy-dependent transactions.

To prevent clients from committing transactions under obsolete policies,
\sys{} additionally requires that a regular transaction is validated against the latest policy version available at its commit timestamp
(Lines~\ref{alg:gov-txn-cc-line:if-latest}--\ref{alg:gov-txn-cc-line:endif-latest}).
This ensures that policy changes take effect immediately in the serialization order and cannot be circumvented by reordering.

\myhdr{Endorsements and Governance} By default, governance transactions are subject to the endorsement requirements of the policies they modify.
Applications may also define alternative endorsement rules for governance transactions ({\em e.g.}, requiring endorsement from a majority of all clients).
Governance transactions are exempt from information-flow checks, as their effects are limited to policy state rather than application data.

\begin{algorithm}[t]
\caption{Governance Transaction OCC}\label{alg:gov-txn-cc}
\begin{algorithmic}[1]
\State // Regular transaction $T$ 
\For{$\mathit{key} \in \writeset{T}$}\label{alg:gov-txn-cc-line:for-add-implicit}
    \State $p, \ \mathit{ts_p} = \mathit{GetPolicy(key, \ ts_T)}$
    \If{$\lnot$ ($ts_p$ is the latest version of $p$)}\label{alg:gov-txn-cc-line:if-latest}
        \State \Return Abort
    \EndIf\label{alg:gov-txn-cc-line:endif-latest}
    \State $\mathit{ImpPolicyReads_T.add((p, \ ts_p))}$
\EndFor\label{alg:gov-txn-cc-line:endfor-add-implicit}
\State
\State // Standard OCC for $T$
\For{$\mathit{(p, \ ts_p) \in ImpPolicyReads_T}$}\label{alg:gov-txn-cc-line:for-regular-vs-gov}
    \If{$\exists G \in \mathit{Committed}: p \in \writeset{G} \land \mathit{ts_p < ts_G < ts_T}$}
        \State \Return Abort
    \EndIf
\EndFor\label{alg:gov-txn-cc-line:endfor-regular-vs-gov}
\State
\State // Standard OCC for governance transaction $G$
\For{$p \in \writeset{G}$}\label{alg:gov-txn-cc-line:for-check-implicit}
    \If{\parbox[t]{0.9\linewidth}{
        $\exists T \in \mathit{Committed: (p, \ ts_p) \in ImpPolicyReads_T}$ \\
        $ \land \mathit{ts_p < ts_G < ts_T}$
    }}
        \State \Return Abort
    \EndIf
\EndFor\label{alg:gov-txn-cc-line:endfor-check-implicit}
\end{algorithmic}
\end{algorithm}

\subsection{Practical Considerations}
\label{subsec:design-practical}




\myhdr{Selecting Validation Clients}
When beginning a transaction, an initiating client selects a set of validation clients $V$ to collect endorsements from. As long as the selected set can pass endorsement validation, the identities of the selected clients are immaterial. Determining $V$'s size, however, is not straightforward. 
Ideally, $|V|$ equals the maximum policy associated with any key in the transaction's write set. 
In practice, however, selecting an appropriate $V$ can be complicated by a limited knowledge of the write set, stale policy information, and by faulty validators. 

\one{} \textit{Write-set discovery.}
For interactive transactions, the full write set may not be known at the start of execution. 
If the initiating client later modifies a key whose policy exceeds $|V|$, it must contact additional validation clients and wait on them to re-execute the transaction. 
In the worst case, this incurs the overhead of an additional execution pass.
Applications can reduce this cost by providing application-specific estimation functions.
For example, if policies are assigned per table and the tables modified by a transaction are known in advance, the initiating client can determine $|V|$ precisely.

\two{} \textit{Stale policy views.}
Clients cache policy information to avoid fetching it on every transaction. 
After a governance transaction commits, this cache may become stale, causing a client to request too few or too many endorsements. 
If too few, the transaction aborts during commit; if too many, the client incurs unnecessary overhead. 
Replicas can return updated policy information as a response. 
Such information must be accompanied by a proof of commit to prevent Byzantine replicas from supplying incorrect policies. 
We assume policy updates are infrequent, limiting the cost of cache refreshes.

\three{} \textit{Faulty validation clients.}
The selected validation set $V$ may include faulty clients that fail to respond or return invalid endorsements. 
If $|V|$ is exactly equal to the required policy threshold, a single faulty validator can prevent progress.
When this occurs, the initiating client may contact additional validators and repeat endorsement collection. 
In practice, clients can track validator responsiveness and reliability over time, biasing future selections toward well-behaved peers. 
Clients may also deliberately overestimate $|V|$ to tolerate occasional failures.


\myhdr{Optimizations}
Several aspects of the \sys{} protocol can be parallelized to reduce latency.
Validation clients may verify forwarded database proofs asynchronously, provided verification completes before endorsement.
Similarly, initiating clients may defer endorsement verification, as replicas re-validate endorsements during commit. 
At replicas, endorsement validation and information-flow checks can be performed in parallel with concurrency control. 
A replica votes to commit only if all checks succeed. 
These optimizations reduce end-to-end latency without weakening safety. 

Finally, rather than allowing clients to choose endorsers independently, applications can enforce a deterministic random selection of endorsers to prevent Byzantine clients from populating their validation sets with other Byzantine clients. Such packing poses no safety risk because validation thresholds are chosen to resist collusion. However, deterministic random selection may enable lower policy thresholds that nonetheless, in expectation, yield safety with high probability. This technique, however, does not reduce the threshold requirements needed for {\em provable} safety.

\section{Evaluation}

Our evaluation addresses the following questions:
\begin{itemize}
    \item What is the overhead \sys{} adds on realistic applications?~(\S\ref{subsec:eval-macro})
    \item How does policy affect \sys{}'s performance?~(\S\ref{subsec:eval-vary-policy})
    \item How responsive is \sys{}'s policy scheme to changes?~(\S\ref{subsec:eval-gov-txn})
    \item What residual impact can Byzantine clients have on \sys{}?~(\S\ref{subsec:eval-byz-clients})
    \item How can lifting be leveraged to enable heterogeneous policies in \sys{} and improve performance?~(\S\ref{subsec:eval-lifting})
    \item How does dynamic write-set discovery impact \sys{}'s performance?~(\S\ref{subsec:eval-ws-discovery})
\end{itemize}

\begin{figure*}[t]
    \begin{minipage}{0.32\textwidth}
        \centering
        \includegraphics[width=\textwidth]{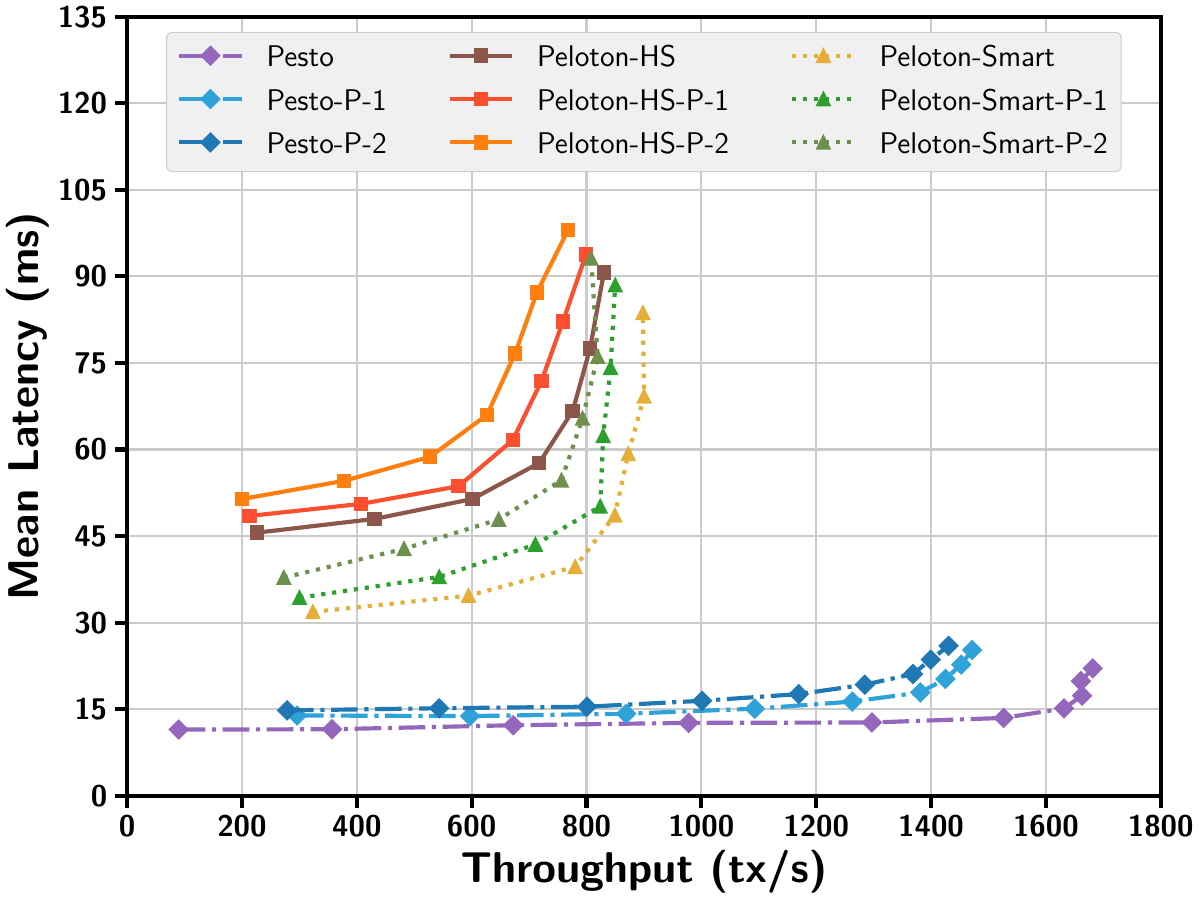}
        \captionof{figure}{\tpcc{} (20 warehouses)}
        \label{fig:tpcc-sql}
    \end{minipage}
    \hfill
    \begin{minipage}{0.32\textwidth}
        \centering
        \includegraphics[width=\textwidth]{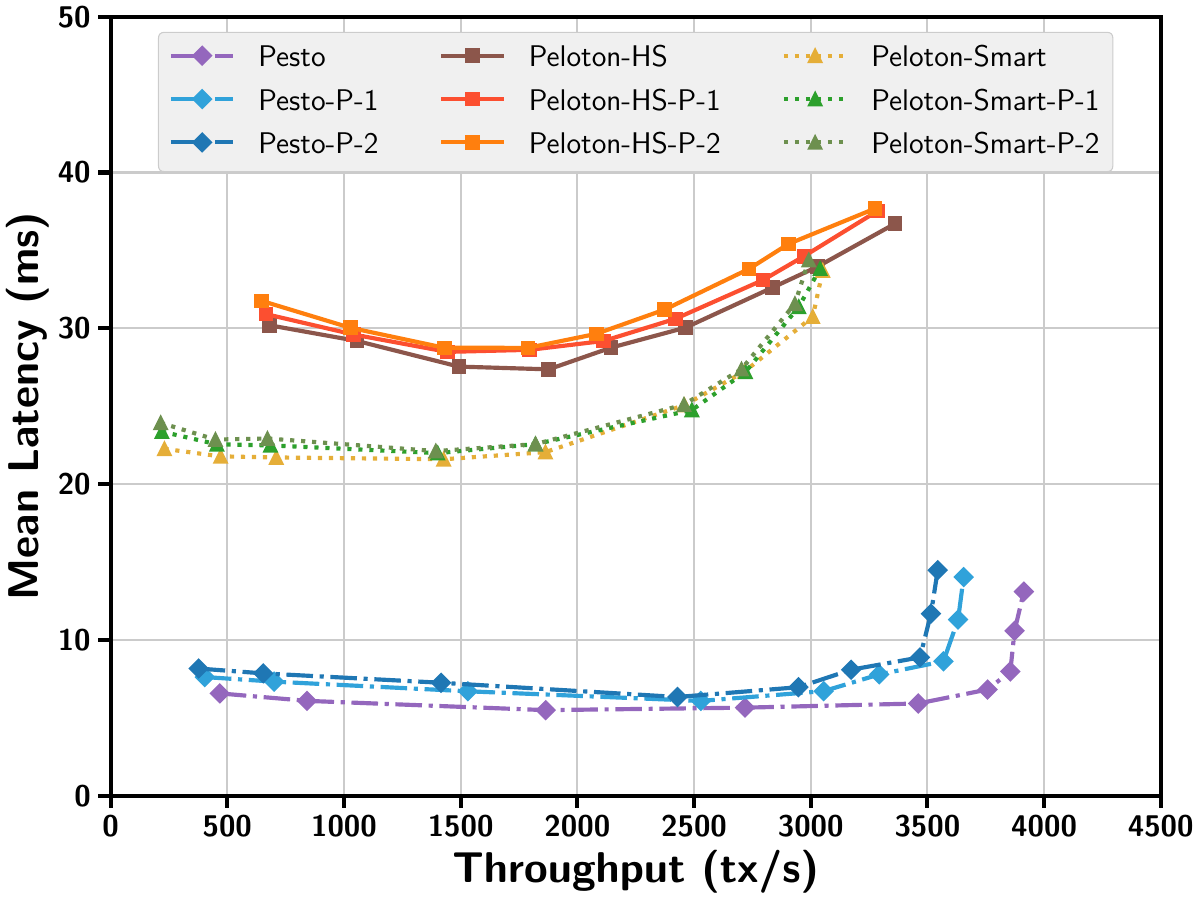}
        \captionof{figure}{\seats{}}
        \label{fig:seats}
    \end{minipage}
    \hfill
    \begin{minipage}{0.32\textwidth}
        \centering
        \includegraphics[width=\textwidth]{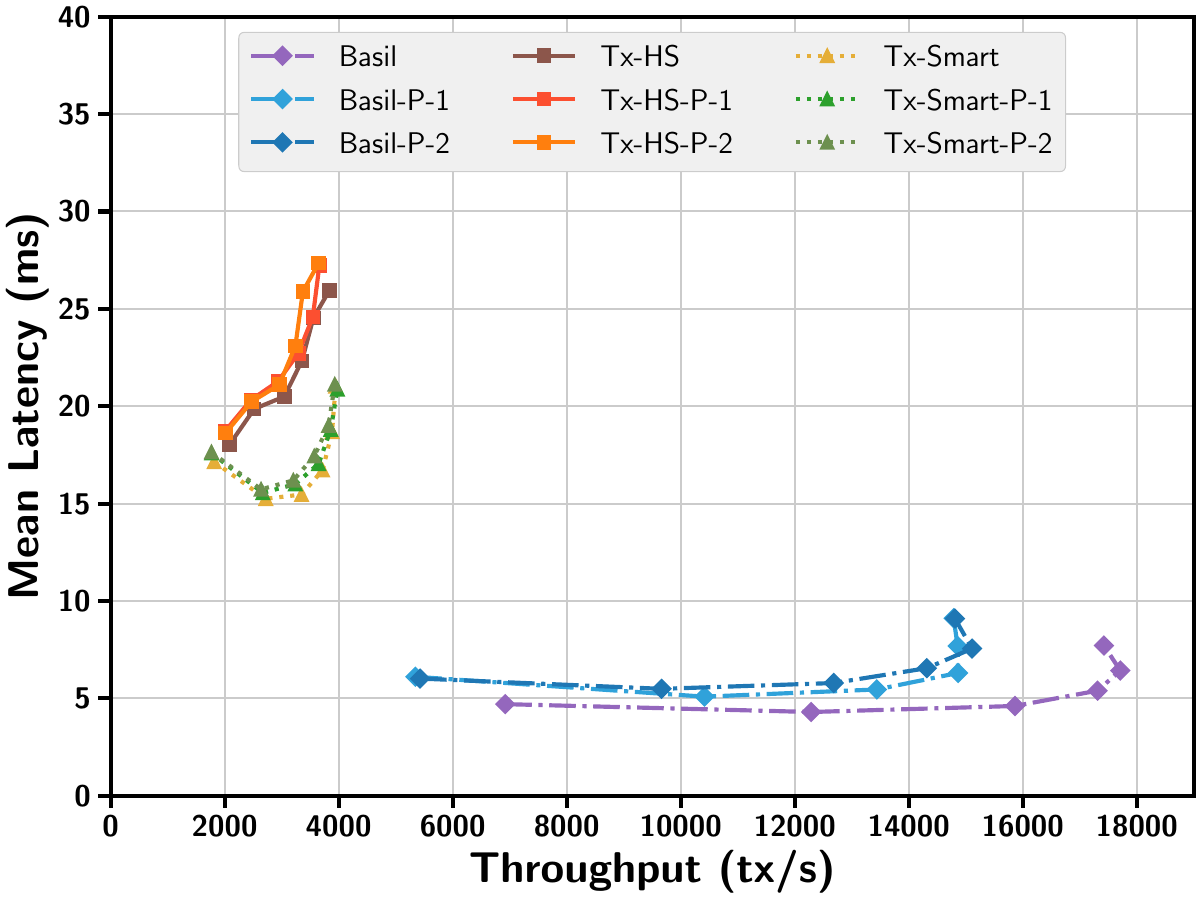}
        \captionof{figure}{\smallbank{}}
        \label{fig:smallbank-kv}
    \end{minipage}
\end{figure*}

\myhdr{Implementation}
We implement \sys{} in C++ atop several transactional BFT systems.
Networking uses Protobuf~\cite{protobuf} over TCP; authentication relies on ed25519 signatures~\cite{ed25519, ed25519-donna} and HMAC-SHA256~\cite{hmac}.
We use Blake3~\cite{blake3} for hashing.

\myhdr{Baselines}
We integrate \sys{} with four systems:\footnote{Code and experimental scripts are available at \url{https://github.com/ATLi2001/Sintr-Artifact}.}

\one{} Basil~\cite{basil}, a sharded BFT key-value store (KVS) with interactive transactions;

\two{} Pesto~\cite{pesto} a BFT SQL database that extends Basil with Peloton's~\cite{peloton} query engine;

\three{} Peloton-SMR, which executes SQL queries at replicas after ordering them via BFT SMR. 
We take the Peloton-SMR implementation from the open source Pesto artifact, which includes two versions: Peloton-HS uses HotStuff~\cite{hotstuff} for consensus; Peloton-Smart instead relies on \bftsmart{}~\cite{bftsmart,bftsmart-github}.

\four{} Tx-SMR, a transactional KVS system layered over BFT SMR. 
We take the Tx-SMR implementation from the open source Basil artifact, which also includes two versions: Tx-HS uses HotStuff as its consensus engine, while Tx-Smart relies on \bftsmart{}.

\myhdr{Experimental Setup}
Experiments run on \texttt{m510} machines (8-core 2.0 GHz CPU, 64 GB RAM, 10 GB NIC, 0.15 ms ping latency) on CloudLab~\cite{cloudlab}.
Clients execute transactions in a closed-loop using a single thread, and retry aborts with exponential back-off.
\sys{} allocates three additional cores per client for validation; these cores do not accelerate transaction execution, which is largely sequential and bottlenecked on asynchronous database calls.
Unless stated otherwise, we assume that \sys{} clients have static knowledge of the required endorsements for all transactions; in Section~\ref{subsec:eval-ws-discovery} we explicitly evaluate the impact of imperfect predictions and dynamic policy discovery.

Each system is configured to tolerate $f=1$ faults ($n=3f+1$ replicas for HotStuff/\bftsmart{}, $n=5f+1$ for Basil/Pesto).
All systems run in-memory and enforce serializable isolation.
Experiments run for 60~s total, including a 15~s warm-up and a 15~s cool-down.

\subsection{High Level Performance}
\label{subsec:eval-macro}

We evaluate \sys{} on three popular transactional benchmarks: 
\tpcc{}~\cite{tpcc}, \seats{}~\cite{oltp-bench}, and \smallbank{}~\cite{oltp-bench}. 
SQL systems run \tpcc{} (20 warehouses, high contention) and \seats{} (lower contention, many read-only queries), while KVS systems run \smallbank{} (1M accounts, 90\% skew to 1{,}000 hot keys).

In each experiment, all objects in the database are assigned the same integrity policy, which remains constant throughout the experiment. 
We denote by \textit{P-x} that a transaction requires $x$ endorsements, excluding the initiating client. 
Baselines correspond to P-0.
Initiating clients select validation clients uniformly in a round-robin manner.

Figures~\ref{fig:tpcc-sql}, \ref{fig:seats}, and \ref{fig:smallbank-kv} report throughput and latency.

\myhdr{\tpcc{}}
Pesto with P-1 and P-2 achieves peak throughput (1471 tx/s and 1431 tx/s, respectively) within 12.5\% and 14.9\% of base Pesto (1681 tx/s).
Latency increases by 16.8\% and 21.7\% at the respective inflection points. 
Peloton-SMR incurs lower overhead: compared with the baseline, throughput with P-1 drops by 6.4\% (Peloton-HS) and 3.0\% (Peloton-Smart), while latency increases, respectively, by 6.9\% and 3.1\%. 

The overheads of \sys{} stem primarily from the additional cryptographic operations required by its protocol. 
At replicas, \sys{} incurs the cost of verifying endorsement signatures, while at clients it adds HMAC computation for forwarded requests and signature verification for forwarded proofs. 
Most validation latency is hidden by \sys{}'s pipelined design and by clients' precise knowledge of the number of validators required. 
As a result, \sys{}'s overhead is more pronounced under Pesto, which is CPU-bound and therefore sensitive to the added signature verification costs. 
In contrast, Peloton-SMR systems are dominated by contention induced by higher baseline latency, making the additional overhead of \sys{} relatively smaller. 

\begin{figure*}[t]
    \begin{minipage}{0.66\textwidth}
        \centering
        \subfloat[Uniform workload\label{fig:vary-policy-u}]{\includegraphics[width=0.48\textwidth]{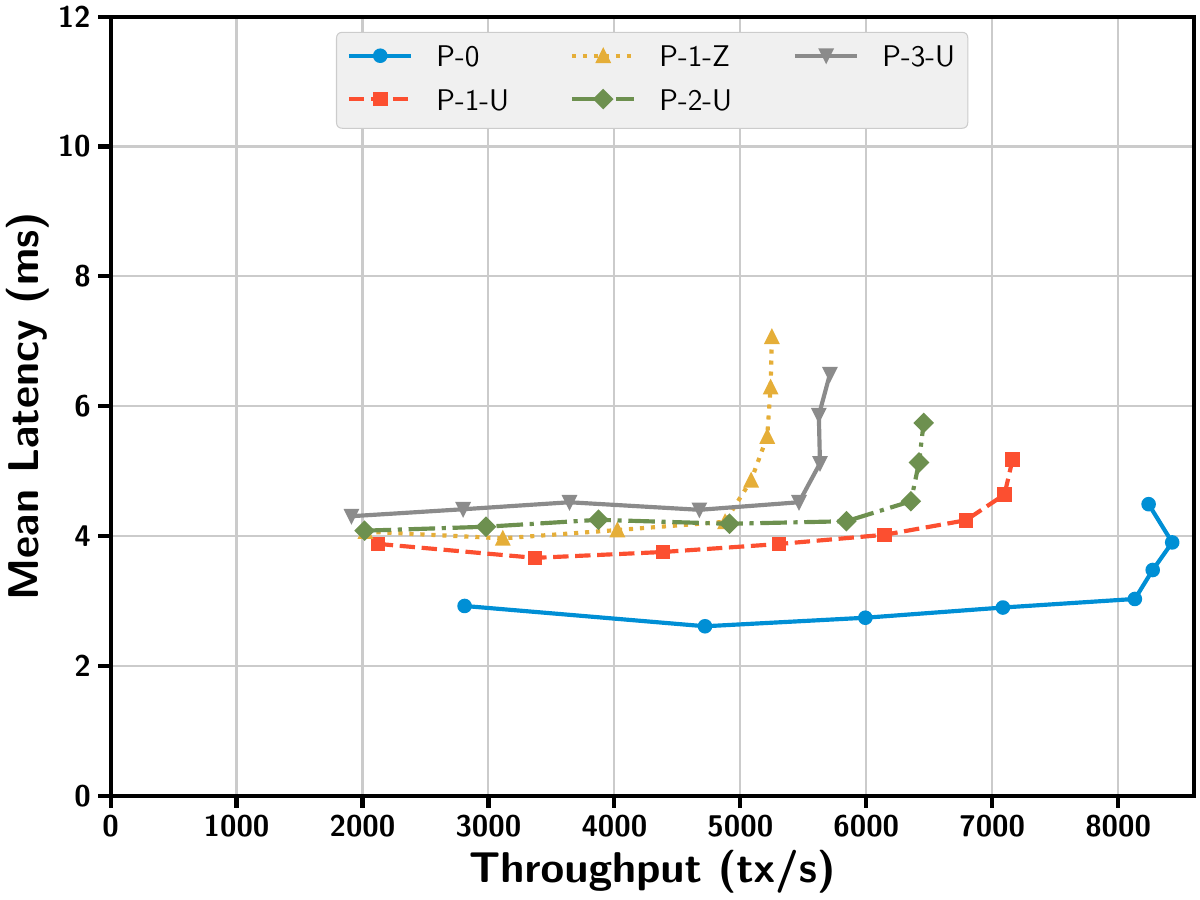}}
        \hfill
        \subfloat[Zipfian workload\label{fig:vary-policy-z}]{\includegraphics[width=0.48\textwidth]{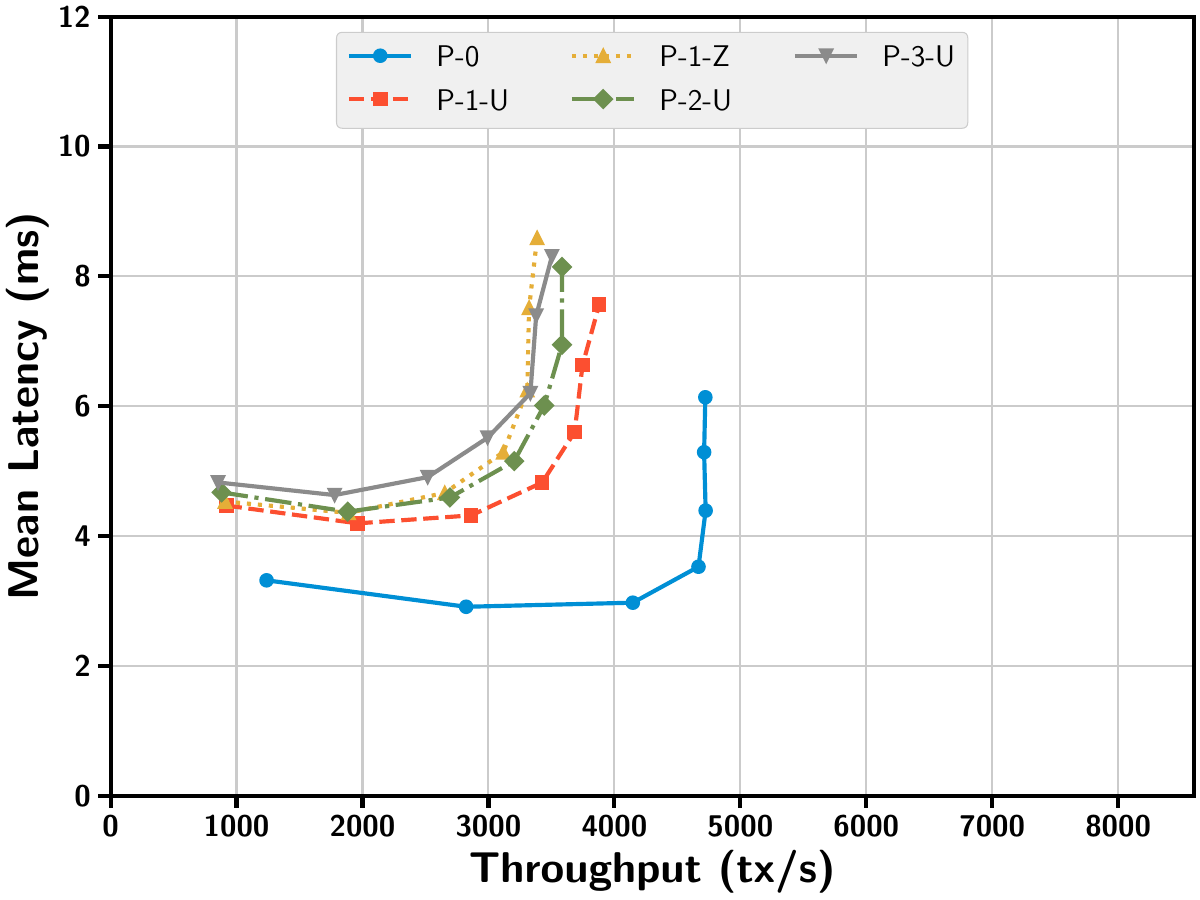}}
        \caption{Impact of application policy}
        \label{fig:vary-policy}
    \end{minipage}
    \hfill
    \begin{minipage}{0.32\textwidth}
        \centering
        \includegraphics[width=\textwidth]{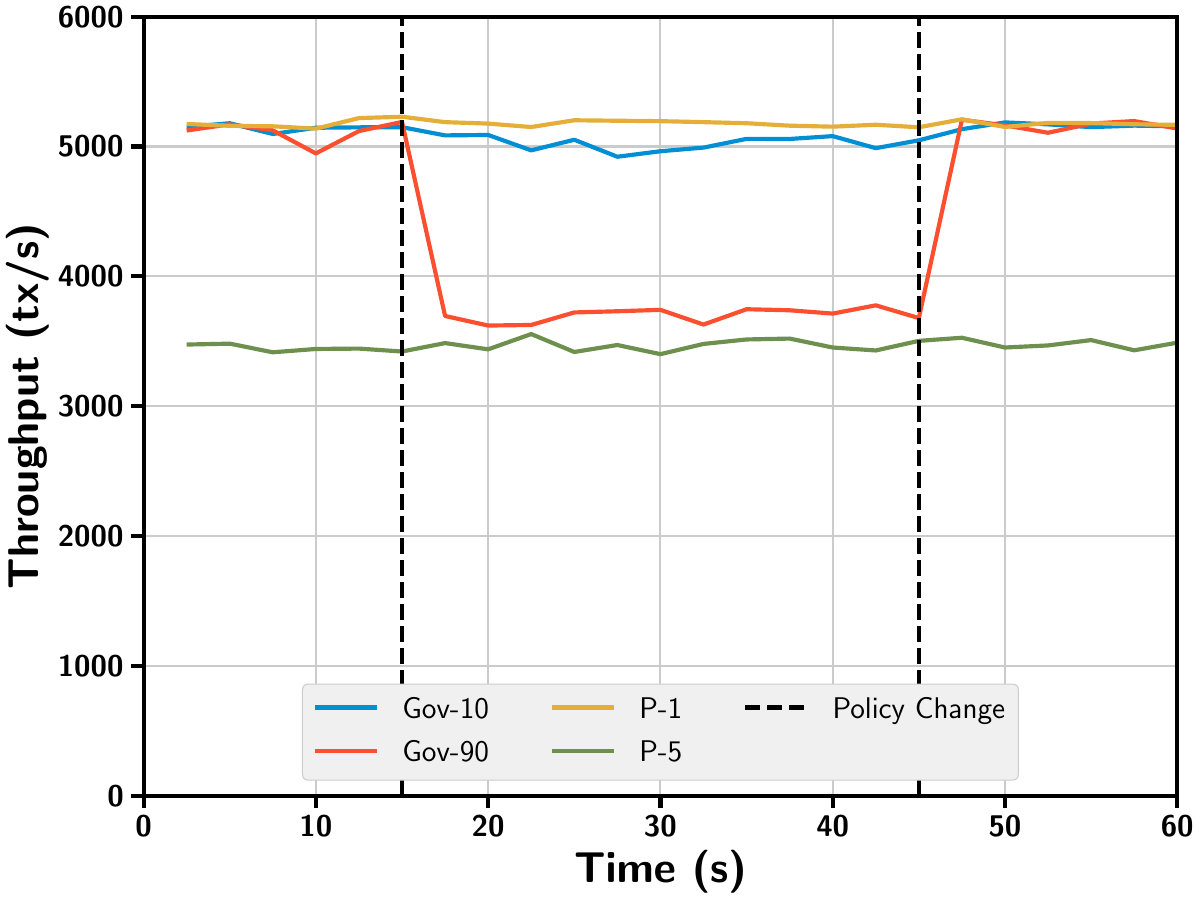}
        \caption{Dynamically changing policies}
        \label{fig:gov-txn}
    \end{minipage}
\end{figure*}

\myhdr{\seats{}}
Pesto with P-1 and P-2 achieve throughputs (3569, 3468 tx/s) that are within 7.4\% and 10.0\% of base Pesto (3855 tx/s), respectively.
The latencies are 8.2\% and 11.3\% higher than base Pesto, respectively.
Compared to \tpcc{}, \seats{} contains a significantly higher portion of read-only transactions, which don't incur \sys{}'s overheads.
All systems are CPU bottlenecked, but \sys{} adds little overhead to the Peloton-SMR systems as the latency of replica endorsement checks is relatively low compared to that of consensus.

\myhdr{\smallbank{}}
Basil under P-1 and P-2 achieves similar throughput, remaining within 16.1\% of base Basil. 
However, these policies increase latency by 16.8\% and 21.4\%, respectively. 
Tx-HS under P-1 experiences a smaller impact, with a 4.6\% reduction in throughput and a 4.9\% increase in latency, while \sys{} adds negligible overhead to the base Tx-Smart system. 
As in \tpcc{}, Basil is CPU-bound, whereas Tx-SMR systems are contention-bound, amplifying \sys{}'s overhead on Basil relative to Tx-SMR.


\myhdr{Takeaway}
In SQL and KVS workloads, \sys{} incurs modest overhead, with higher costs primarily in CPU-bound systems.

\subsection{Impact of Application Policy}
\label{subsec:eval-vary-policy}

We study how performance scales with policy strength using a YCSB-based microbenchmark (10 tables, 1M keys each; transactions read and update five rows) using Pesto as a baseline. 
We consider two workloads: an uncontended uniform access pattern and a contended Zipfian access pattern with coefficient 0.99.
For each workload, we instantiate \sys{} under a family of policies denoted \textit{P-x-$[U/Z]$}, where all keys share a static policy, $x$ is the number of required client endorsements, and $U$ and $Z$ indicate uniform and Zipfian distributions of validation load, respectively. 
The Zipfian case models scenarios where certain clients are preferred as validators, {\em e.g.}, because of reputation or proximity. 

Under the uniform workload (Fig.~\ref{fig:vary-policy-u}), \sys{} with P-1-U incurs a 15\% throughput loss and a 32.6\% latency increase relative to P-0. 
Because transactions are short, \sys{}'s overhead is proportionally large, and at high load servers become CPU-bound. 
Strengthening the policy further requires additional signature validations per transaction, reducing throughput by a steady 8--9\% per added endorsement. 
With P-1-Z, skewed validator selection causes a small subset of clients to become CPU-bound on validation.

Under the Zipfian workload (Fig.~\ref{fig:vary-policy-z}), all configurations are contention-bound. 
The added latency from endorsement checks slightly increases transaction conflict windows and abort rates, leading P-1-U to incur a 20.6\% throughput reduction and a 36.8\% latency increase. 
However, because servers are no longer CPU-bound, the marginal cost of additional endorsement checks is small. 
Skewed validation performs only slightly worse, as contention caps throughput well below client validation capacity.

\myhdr{Takeaway}  
\sys{}'s overhead is primarily determined by where the system bottlenecks: it is more pronounced in CPU-bound settings with short transactions, but modest when contention dominates.
Even under skewed or stronger policies, endorsement costs scale predictably.

\subsection{Dynamically Changing Policies}
\label{subsec:eval-gov-txn}

We evaluate dynamic policy updates in \sys{} using a YCSB-based microbenchmark with a uniform access pattern. 
The experiment runs for 90~s with 20 clients, including 15~s warm-up and cool-down periods. 
Initially, all keys are assigned P-1. 
We then inject a governance transaction that upgrades $x\%$ of keys (Gov-$x$) from P-1 to P-5, followed 30~s later by a second governance transaction that reverts the change. 
For reference, we also report steady-state performance with all keys fixed at P-1 or P-5.

Figure~\ref{fig:gov-txn} illustrates the effect of policy updates. 
Governance transactions commit in under 5~ms, on par with regular transactions. 
When the first update commits, each client aborts at most one in-flight transaction due to concurrency control or insufficient endorsements. 
After the second update, clients temporarily generate excess endorsements, but detect and correct this on their next transaction. 
Throughout, transaction processing never halts, and throughput shifts immediately to the level dictated by the active policy.

\myhdr{Takeaway} 
\sys{} supports fast, minimally disruptive policy changes at runtime. 
They take effect immediately, with throughput adapting smoothly and without pausing execution.

\begin{figure*}[t]
    \begin{minipage}{0.66\textwidth}
        \centering
        \subfloat[Uniform workload\label{fig:byz-clients-u}]{\includegraphics[width=0.48\textwidth]{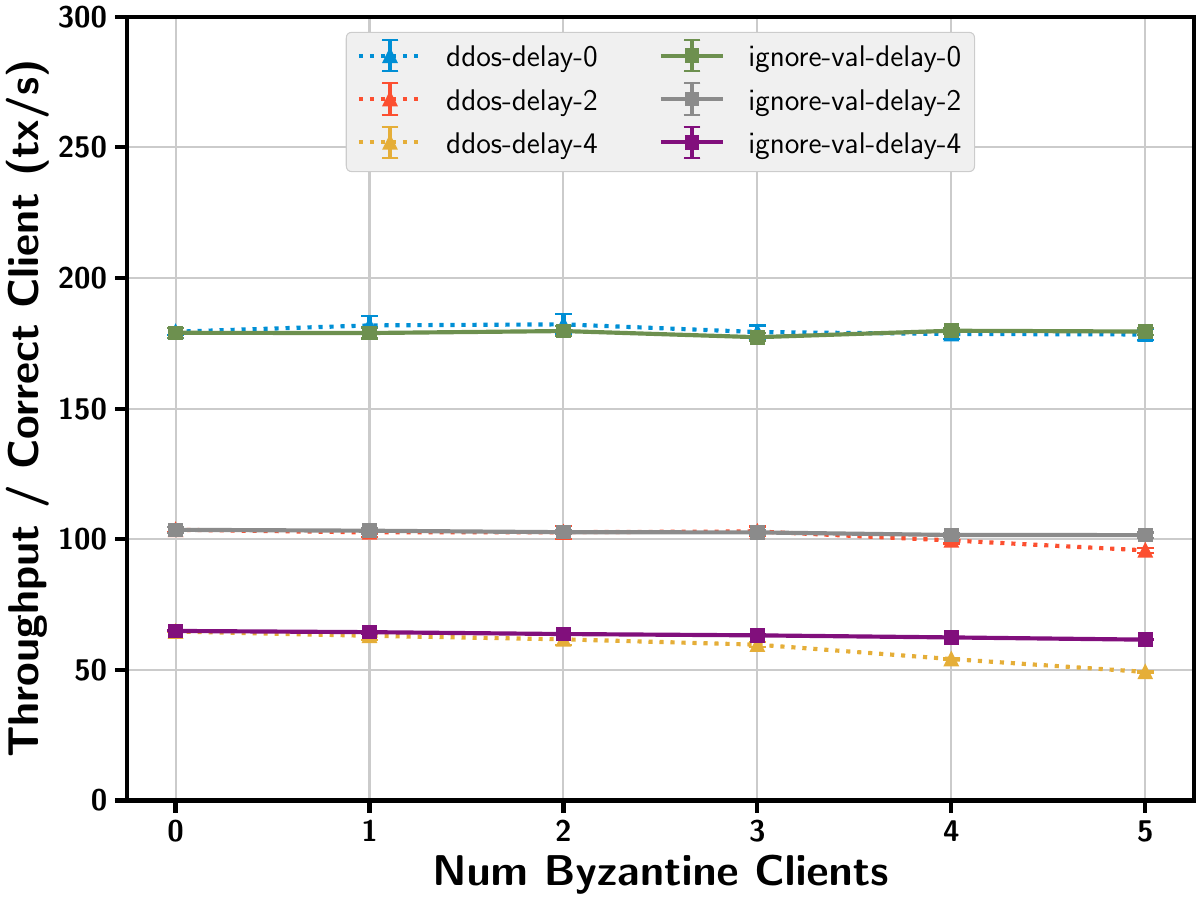}}
        \hfill
        \subfloat[Zipfian workload\label{fig:byz-clients-z}]{\includegraphics[width=0.48\textwidth]{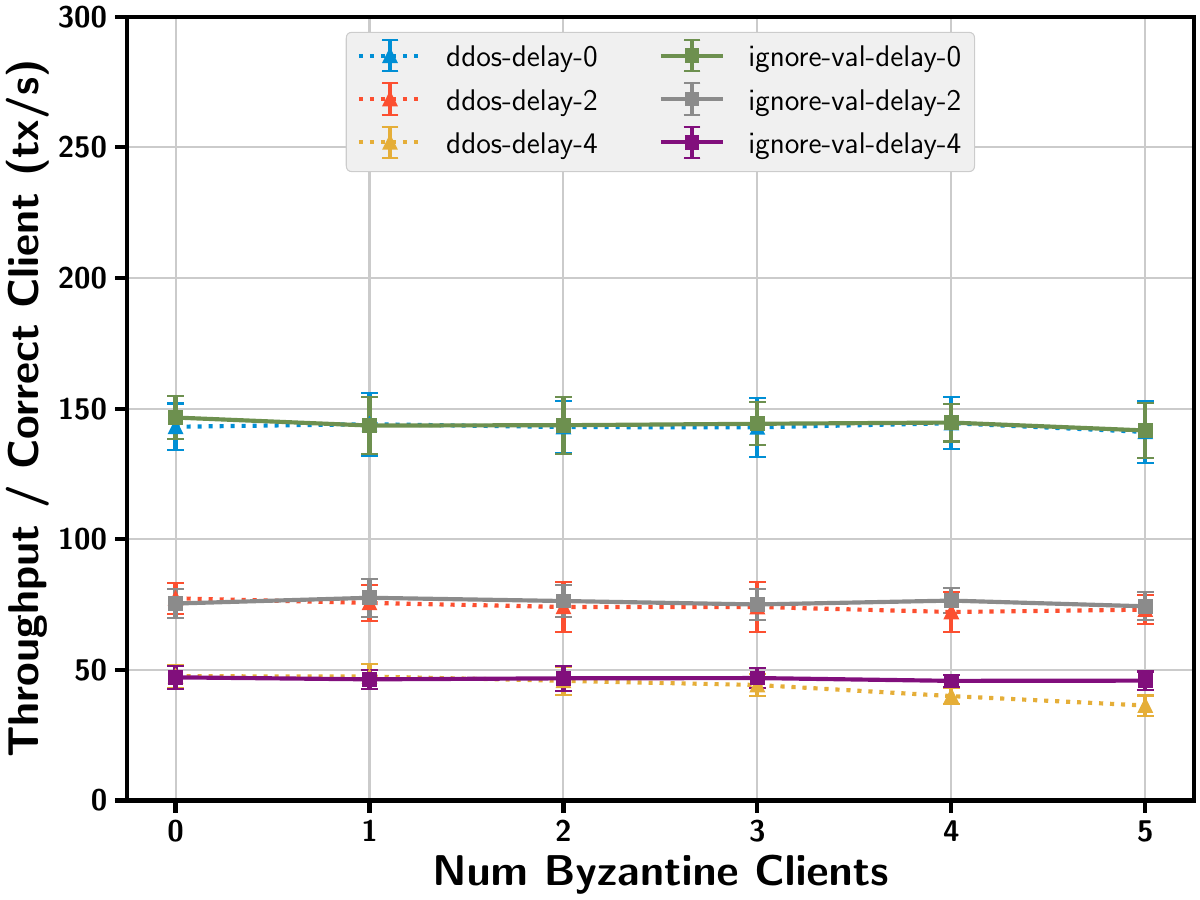}}
        \caption{Residual impact of Byzantine clients}
        \label{fig:byz-clients}
    \end{minipage}
    \hfill
    \begin{minipage}{0.32\textwidth}
        \centering
        \includegraphics[width=\textwidth]{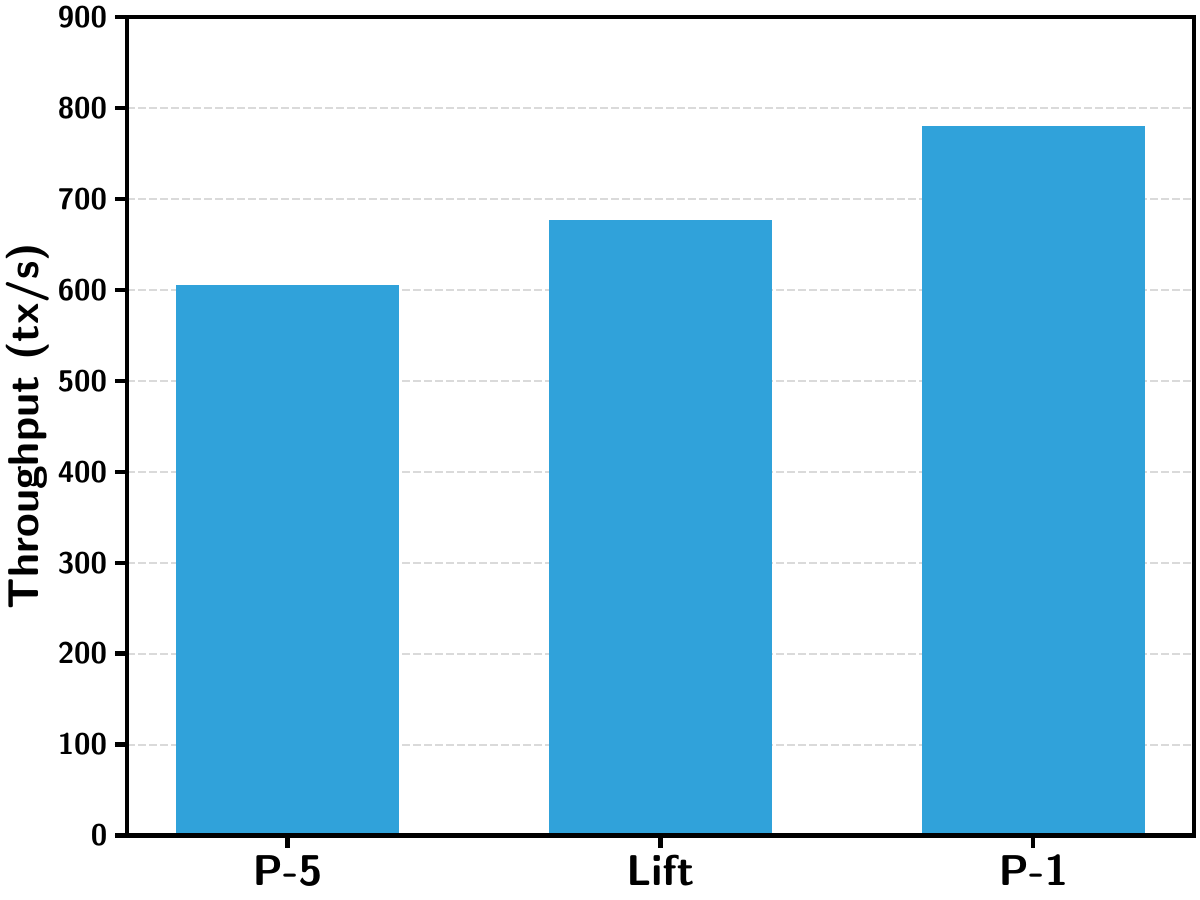}
        \caption{Leveraging policy lifting}
        \label{fig:lifting-evaluation}
    \end{minipage}
\end{figure*}

\subsection{Residual Impact of Byzantine Clients}
\label{subsec:eval-byz-clients}

\sys{}'s policy-based design ensures that, under a well-chosen policy, Byzantine clients cannot compromise execution correctness. 
Nevertheless, Byzantine clients may still degrade performance by deviating from the endorsement protocol---for example, by ignoring validation requests or by soliciting excessive endorsements for their own transactions. 
Such behavior is indistinguishable from benign behavior: an unresponsive client may simply pose as slow, and correct clients cannot tell whether their endorsement was strictly necessary.

We study two representative attacks: \textit{ignore-val}, where Byzantine clients ignore validation requests, and \textit{ddos}, where they request endorsements from all clients. 
Other denial-of-service behaviors ({\em e.g.}, flooding the system with arbitrary validation requests) are mitigated by bounding the number of outstanding transactions per client.

Our evaluation uses a YCSB-based microbenchmark with uniform and Zipfian access patterns. 
We deploy 20 clients and configure the system to tolerate up to five Byzantine clients by assigning P-5 to all keys. 
Validation cost is varied using an artificial $x$~ms busy-wait (\textit{delay-x}). 
Correct clients optimistically contact five validators per transaction; in \textit{ignore-val}, ignored requests trigger contacting additional validators, both immediately and in subsequent transactions.

Figure~\ref{fig:byz-clients} reports throughput of correct clients. 
Under \textit{ignore-val}, performance degradation is small (1--5\%) across all validation costs. 
Timeouts allow correct clients to identify unresponsive validators, requiring only 1--2 additional contacts on average, which validation clients can absorb without saturating.


In contrast, \textit{ddos} attacks grow more disruptive as validation becomes more expensive, ranging from negligible impact at \textit{delay-0} to a 25\% throughput reduction at \textit{delay-4}. 
When validation is cheap, correct clients can handle surplus requests; as validation cost grows, excess endorsement work delays transaction processing. 
Because this pressure is primarily client-side, the effect is similar under uniform and Zipfian workloads.


In practice, applications may restrict which peers clients contact for validation. 
For low-compute transactions, such policies can remain permissive, as \sys{} retains high throughput even under adversarial validation load.

\myhdr{Takeaway} 
Byzantine clients in \sys{} cannot compromise correctness and have limited performance impact unless validation is computationally expensive. 


\subsection{Leveraging Policy Lifting}
\label{subsec:eval-lifting}


\sys{} allows applications to define heterogeneous policies so they pay high safety costs only where needed. 
For example, a \tpcc{}-like application (Fig.~\ref{fig:app-policy}) manages high-integrity financial data alongside lower-integrity record-keeping data, and ideally enforces strong policies only on the former.

However, the \texttt{Delivery} transaction reads low-policy order data while updating high-policy financial state, violating safe information flow. 
To resolve this, \sys{} supports lifting, allowing transactions to promote low-policy objects after validation. 
In \texttt{Delivery}, a lift function checks orders against customer balances; upon success, the order data is lifted  and the transaction proceeds safely. 
This preserves heterogeneous policies without sacrificing correctness.



We evaluate lifting using a modified \tpcc{} workload in which \texttt{Delivery} executes a lift function that checks that order amounts do not exceed available credit. 
The experiment uses 20 clients, tolerating up to five Byzantine clients. 
Financial tables (warehouse, district, customer, history) are assigned P-5, while record-keeping tables (orders, stock, item) are assigned P-1. 
We compare this configuration against uniform P-5 and uniform P-1 deployments. 
The uniform P-5 setup incurs unnecessary overhead on record‑keeping data, while uniform P-1 accepts undue risk for financial data.

Figure~\ref{fig:lifting-evaluation} shows the resulting throughput. 
Lifting improves throughput by 11.6\% over uniform P-5.
Lifting results in marginally more transactions requiring five endorsements than one endorsement; it proportionally recovers 43\% of the throughput gained by the unsafe P-1 configuration while preserving P-5 safety for financial data. 

\myhdr{Takeaway} 
Lifting enables heterogeneous policies that substantially improve performance without sacrificing safety.

\subsection{Impact of Dynamic Write-set Discovery}
\label{subsec:eval-ws-discovery}

\begin{figure}[t!]
    \centering
    \includegraphics[width=0.67\columnwidth]{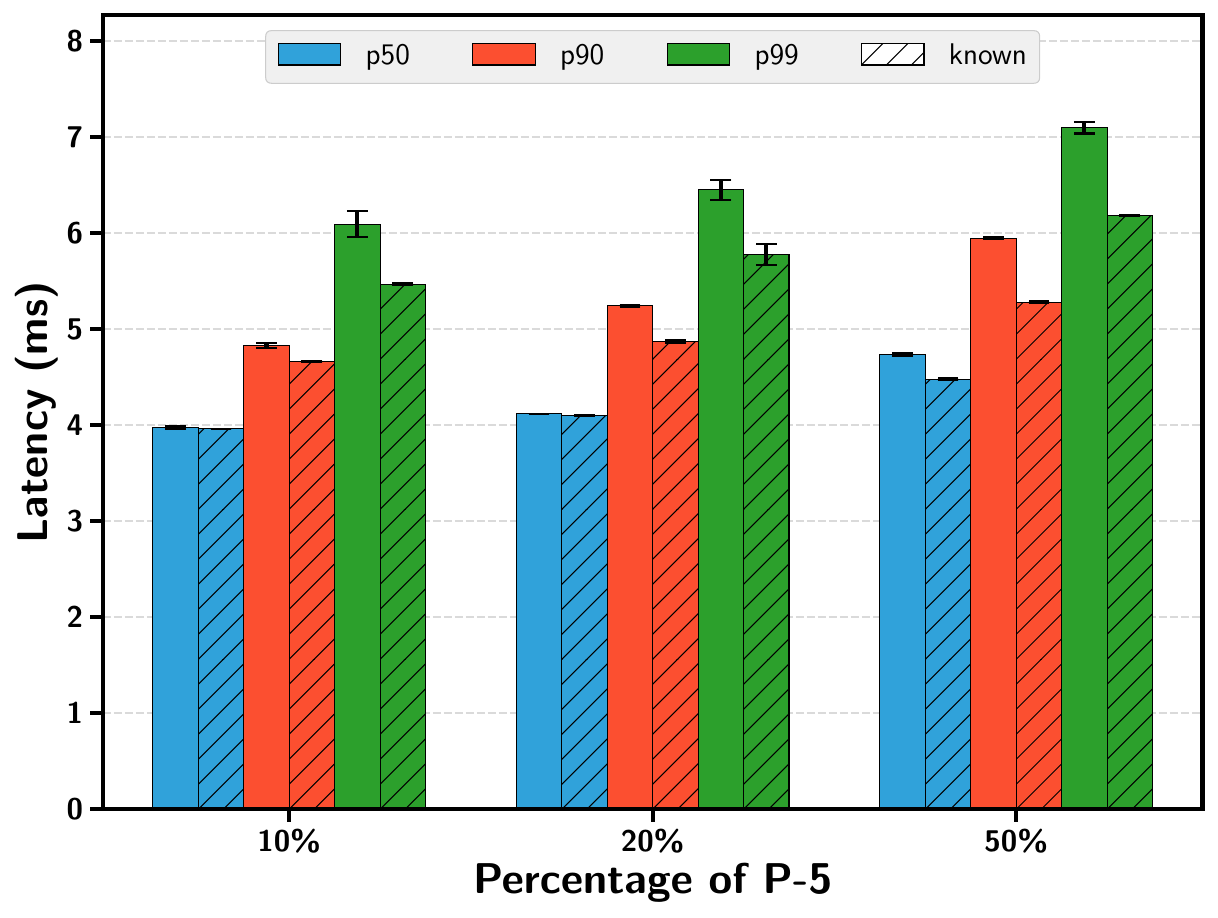}
    \caption{Dynamic write-set discovery}
    \label{fig:ws-discovery}
\end{figure}

\sys{} does not require advance knowledge of a transaction's write set and can dynamically infer during execution the size of the set $V$ of validation clients required by the transaction. If the client initiating a transaction $T$ modifies a key whose policy exceeds the client's initial estimate for $|V|$, it must contact more validation clients and wait on them to re-execute $T$. If this step happens early in the $T$'s execution, it causes little additional end-to-end latency, as the processing  of these extra requests overlaps with the transaction's ongoing execution.

The worst case is when the key update that triggers the need for a larger $V$ is $T$'s last: this adds one message delay to contact the additional validation clients plus the time those clients need to execute $T$---which they can do without issuing   additional remote database calls. 

We evaluate the impact of dynamic write-set discovery by varying the fraction of P-5 versus P-1 keys using a YCSB-based microbenchmark with a uniform access pattern. 
The experiment uses 20 clients and each transaction reads and updates five rows.
Clients always begin transactions by contacting one validation client.
If a client's transaction updates a P-5 key, then it dynamically contacts four additional validators.

Figure~\ref{fig:ws-discovery} shows the latency under dynamic discovery, compared with that obtained with initial perfect knowledge of $V$ (shaded).
As the percentage of P-5 keys increases, transactions are more likely to be delayed by waiting on four additional validation clients.
The resulting delay is primarily due to \sys{}'s cryptographic operations, as message delays and local transaction execution times are small.
The p99 latency also reflects the worst-case for dynamic discovery when the last update is to a P-5 key; at 50\% P-5 keys, the delay caused by lack of perfect initial knowledge is about 1 ms.


\myhdr{Takeaway} 
\sys{} remains efficient without advance knowledge of the write set. Dynamic discovery incurs only modest latency when policies are mostly homogeneous, with the overhead increasing to roughly a single message round trip as policy heterogeneity grows.

    
\section{Related Work}

\myhdr{BFT State Machine Replication}
State machine replication (SMR)~\cite{smr} provides the abstraction of a single fault-tolerant server and underpins a wide range of distributed storage systems in both crash fault-tolerant (CFT)~\cite{spanner,cockroachdb} and Byzantine fault-tolerant (BFT) settings~\cite{hyperledger-fabric,tendermint,librabft,omniledger}.
BFT SMR systems rely on consensus protocols~\cite{pbft,zyzzyva,hotstuff,aardvark,autobahn,sbft} to ensure replicas agree on a total order of requests despite arbitrary failures.
A common approach is to execute application logic at the replicas themselves, thereby ensuring that all committed state transitions arise from correct execution of application-defined commands.
Blockchains~\cite{ethereum,solana,avalanche} follow this model, often by embedding application logic in custom smart-contract languages.
However, colocating application execution with storage imposes significant usability costs, including reliance on domain-specific languages and limited ability to scale application logic independently of the backend~\cite{pavlo-sigmod-keynote}.

\myhdr{Client-Side Computation in Byzantine Systems} Several BFT systems improve performance and developer usability by shifting application execution to clients.
Client-centric transactional systems such as Basil, Pesto, and HRDB~\cite{basil,pesto,hrdb} expose database-style interfaces while relying on replicas to enforce ordering and isolation.
Similarly, client-driven Byzantine storage and quorum systems~\cite{byzantine-quorum-systems,q/u,HQ} delegate computation to clients while replicas ensure consistency over stored values or updates.
While these designs provide strong guarantees at the storage or consistency level, they share a common limitation: replicas do not validate the semantic correctness of client-side computation.
Authenticated Byzantine clients may therefore issue operations that are serializable or quorum-certified yet violate application-level invariants.
When such systems are used as substrates for transactional or application-level services, they relinquish the guarantee that committed updates arise from correct executions of application logic.
\sys{} addresses this architectural gap by restoring application execution validity in systems that permit client-side interactive computation, without reverting to replica-side execution.

\myhdr{Transaction Validation by Redundant Execution}
Prior systems~\cite{hyperledger-fabric,chainify-db,eve,nysiad} employ redundant execution to validate transaction outcomes.
However, these approaches retain the SMR execution model and require application logic to execute within the replica, limiting scalability and flexibility.
Moreover, redundant execution typically occurs over replica-local state, allowing even correct replicas to diverge.
In contrast, \sys{} uses redundant execution among clients to validate interactive transactions while ensuring that all correct validators reach identical outcomes.

\myhdr{Database Integrity Constraints} Modern databases support integrity constraints that automatically enforce semantic properties of stored data ({\em e.g.}, non-negativity and referential integrity)~\cite{db-sys-textbook,postgres}. Although these constraints remain valuable in \sys{}, enforcing constraints powerful enough to check whether a transaction's execution is valid quickly degenerates into re-executing the transaction within the database. For example, even for a simple payment transaction (Figure~\ref{fig:boundary}), verifying that credits equal debits requires examining the transaction inputs together with the initial and final account balances, effectively re-executing the transaction. Instead, \sys{} keeps transaction execution entirely client-side while still ensuring valid application executions.

\myhdr{Information Flow}
\sys{} draws inspiration from classical information flow control (IFC) techniques~\cite{noninterference,biba-integrity,language-based-ifc}. 
Language-based systems~\cite{fabric,jflow,ifdb} enforce confidentiality and integrity by ensuring data flows respect security policies. 
These approaches address a different problem: they regulate \textit{who} may access or modify data, but not \textit{how} application logic is executed. 
A node with sufficient privileges may still perform incorrect application logic while respecting IFC constraints. 
\sys{}, in contrast, focuses on ensuring that transactions are executed according to their application-defined semantics, regardless of which clients are authorized to execute them.



\section{Conclusion}

This paper presents \sys{}, a framework for providing fault tolerance against Byzantine clients in client-centric, interactive transactional systems. 
\sys{} makes it possible to reconcile two seemingly conflicting goals: the scalability and expressiveness of client-side transactions, which can benefit from mature programming
languages, tooling, and development workflows; and the strong safety
guarantees of server-side validation. 

\sys{} achieves this reconciliation by leveraging redundant client execution and explicit
endorsement of transaction outcomes, ensuring that committed transactions adhere to application-level integrity policies even in the presence of malicious clients. 
A key contribution of \sys{} is its support for flexible, heterogeneous policies that allow applications to express fine-grained integrity requirements without modifying underlying database protocols. 
By integrating endorsement validation and policy enforcement into the commit path, \sys{} enables existing BFT transactional systems to safely support interactive, client-driven execution with modest overhead.


\begin{acks}
    We are grateful to our shepherd and the anonymous SOSP reviewers for their thorough and insightful comments.
    This work was supported in part by the NSF grant CNS-CORE 2106954.
\end{acks}

\bibliographystyle{ACM-Reference-Format}
\bibliography{references}

@inproceedings{basil,
    author = {Suri-Payer, Florian and Burke, Matthew and Wang, Zheng and Zhang, Yunhao and Alvisi, Lorenzo and Crooks, Natacha},
    title = {Basil: Breaking up BFT with ACID (transactions)},
    year = {2021},
    isbn = {9781450387095},
    publisher = {Association for Computing Machinery},
    address = {New York, NY, USA},
    url = {https://doi.org/10.1145/3477132.3483552},
    doi = {10.1145/3477132.3483552},
    booktitle = {Proceedings of the ACM SIGOPS 28th Symposium on Operating Systems Principles},
    pages = {1--17},
    numpages = {17},
    location = {Virtual Event, Germany},
    series = {SOSP '21}
}

@inproceedings{pesto,
    author = {Suri-Payer, Florian and Giridharan, Neil and Arzola, Liam and Cohen, Shir and Alvisi, Lorenzo and Crooks, Natacha},
    title = {Pesto: Cooking up High Performance BFT Queries},
    year = {2025},
    isbn = {9798400718700},
    publisher = {Association for Computing Machinery},
    address = {New York, NY, USA},
    url = {https://doi.org/10.1145/3731569.3764799},
    doi = {10.1145/3731569.3764799},
    booktitle = {Proceedings of the ACM SIGOPS 31st Symposium on Operating Systems Principles},
    pages = {529--554},
    numpages = {26},
    location = {Lotte Hotel World, Seoul, Republic of Korea},
    series = {SOSP '25}
}

@inproceedings{hyperledger-fabric,
    author = {Androulaki, Elli and Barger, Artem and Bortnikov, Vita and Cachin, Christian and Christidis, Konstantinos and De Caro, Angelo and Enyeart, David and Ferris, Christopher and Laventman, Gennady and Manevich, Yacov and Muralidharan, Srinivasan and Murthy, Chet and Nguyen, Binh and Sethi, Manish and Singh, Gari and Smith, Keith and Sorniotti, Alessandro and Stathakopoulou, Chrysoula and Vukoli\'{c}, Marko and Cocco, Sharon Weed and Yellick, Jason},
    title = {Hyperledger fabric: a distributed operating system for permissioned blockchains},
    year = {2018},
    isbn = {9781450355841},
    publisher = {Association for Computing Machinery},
    address = {New York, NY, USA},
    url = {https://doi.org/10.1145/3190508.3190538},
    doi = {10.1145/3190508.3190538},
    booktitle = {Proceedings of the Thirteenth EuroSys Conference},
    articleno = {30},
    numpages = {15},
    location = {Porto, Portugal},
    series = {EuroSys '18}
}

@inproceedings {pbft,
    author = {Miguel Castro and Barbara Liskov},
    title = {Practical Byzantine Fault Tolerance},
    booktitle = {3rd Symposium on Operating Systems Design and Implementation (OSDI 99)},
    year = {1999},
    address = {New Orleans, LA},
    url = {https://www.usenix.org/conference/osdi-99/practical-byzantine-fault-tolerance},
    publisher = {USENIX Association},
    month = feb
}

@article{zyzzyva,
    author = {Kotla, Ramakrishna and Alvisi, Lorenzo and Dahlin, Mike and Clement, Allen and Wong, Edmund},
    title = {Zyzzyva: Speculative Byzantine fault tolerance},
    year = {2010},
    issue_date = {December 2009},
    publisher = {Association for Computing Machinery},
    address = {New York, NY, USA},
    volume = {27},
    number = {4},
    issn = {0734-2071},
    url = {https://doi.org/10.1145/1658357.1658358},
    doi = {10.1145/1658357.1658358},
    journal = {ACM Trans. Comput. Syst.},
    month = jan,
    articleno = {7},
    numpages = {39}
}

@inproceedings{byzantine-quorum-systems,
    author = {Malkhi, Dahlia and Reiter, Michael},
    title = {Byzantine quorum systems},
    year = {1997},
    isbn = {0897918886},
    publisher = {Association for Computing Machinery},
    address = {New York, NY, USA},
    url = {https://doi.org/10.1145/258533.258650},
    doi = {10.1145/258533.258650},
    booktitle = {Proceedings of the Twenty-Ninth Annual ACM Symposium on Theory of Computing},
    pages = {569--578},
    numpages = {10},
    location = {El Paso, Texas, USA},
    series = {STOC '97}
}

@inproceedings{pavlo-sigmod-keynote,
    author = {Pavlo, Andrew},
    title = {What Are We Doing With Our Lives? Nobody Cares About Our Concurrency Control Research},
    year = {2017},
    isbn = {9781450341974},
    publisher = {Association for Computing Machinery},
    address = {New York, NY, USA},
    url = {https://doi.org/10.1145/3035918.3056096},
    doi = {10.1145/3035918.3056096},
    booktitle = {Proceedings of the 2017 ACM International Conference on Management of Data},
    pages = {3},
    numpages = {1},
    location = {Chicago, Illinois, USA},
    series = {SIGMOD '17}
}

@article{oltp-bench,
    author = {Difallah, Djellel Eddine and Pavlo, Andrew and Curino, Carlo and Cudre-Mauroux, Philippe},
    title = {OLTP-Bench: an extensible testbed for benchmarking relational databases},
    year = {2013},
    issue_date = {December 2013},
    publisher = {VLDB Endowment},
    volume = {7},
    number = {4},
    issn = {2150-8097},
    url = {https://doi.org/10.14778/2732240.2732246},
    doi = {10.14778/2732240.2732246},
    journal = {Proc. VLDB Endow.},
    month = dec,
    pages = {277--288},
    numpages = {12}
}

@inproceedings{cloudlab,
    title = "The Design and Operation of {CloudLab}",
    author = "Dmitry Duplyakin and Robert Ricci and Aleksander Maricq and Gary Wong and Jonathon Duerig and Eric Eide and Leigh Stoller and Mike Hibler and David Johnson and Kirk Webb and Aditya Akella and Kuangching Wang and Glenn Ricart and Larry Landweber and Chip Elliott and Michael Zink and Emmanuel Cecchet and Snigdhaswin Kar and Prabodh Mishra",
    booktitle = "Proceedings of the {USENIX} Annual Technical Conference (ATC)",
    pages = "1--14",
    year = 2019,
    month = jul,
    url = "https://www.flux.utah.edu/paper/duplyakin-atc19"
}

@inproceedings{fabric,
    author = {Liu, Jed and George, Michael D. and Vikram, K. and Qi, Xin and Waye, Lucas and Myers, Andrew C.},
    title = {Fabric: a platform for secure distributed computation and storage},
    year = {2009},
    isbn = {9781605587523},
    publisher = {Association for Computing Machinery},
    address = {New York, NY, USA},
    url = {https://doi.org/10.1145/1629575.1629606},
    doi = {10.1145/1629575.1629606},
    booktitle = {Proceedings of the ACM SIGOPS 22nd Symposium on Operating Systems Principles},
    pages = {321--334},
    numpages = {14},
    location = {Big Sky, Montana, USA},
    series = {SOSP '09}
}

@inproceedings{peerreview,
    author = {Haeberlen, Andreas and Kouznetsov, Petr and Druschel, Peter},
    title = {PeerReview: practical accountability for distributed systems},
    year = {2007},
    isbn = {9781595935915},
    publisher = {Association for Computing Machinery},
    address = {New York, NY, USA},
    url = {https://doi.org/10.1145/1294261.1294279},
    doi = {10.1145/1294261.1294279},
    booktitle = {Proceedings of Twenty-First ACM SIGOPS Symposium on Operating Systems Principles},
    pages = {175--188},
    numpages = {14},
    location = {Stevenson, Washington, USA},
    series = {SOSP '07}
}

@inproceedings{iaccf,
    author = {Alex Shamis and Peter Pietzuch and Burcu Canakci and Miguel Castro and Cedric Fournet and Edward Ashton and Amaury Chamayou and Sylvan Clebsch and Antoine Delignat-Lavaud and Matthew Kerner and Julien Maffre and Olga Vrousgou and Christoph M. Wintersteiger and Manuel Costa and Mark Russinovich},
    title = {{IA-CCF}: Individual Accountability for Permissioned Ledgers},
    booktitle = {19th USENIX Symposium on Networked Systems Design and Implementation (NSDI 22)},
    year = {2022},
    isbn = {978-1-939133-27-4},
    address = {Renton, WA},
    pages = {467--491},
    url = {https://www.usenix.org/conference/nsdi22/presentation/shamis},
    publisher = {USENIX Association},
    month = apr
}

@inproceedings{chainify-db,
  author = {Felix Martin Schuhknecht and Ankur Sharma and Jens Dittrich and Divya Agrawal},
  title = {chainifyDB: How to get rid of your Blockchain and use your DBMS instead},
  booktitle = {11th Annual Conference on Innovative Data Systems Research (CIDR '21)},
  year = {2021}
}

@inproceedings{smallQS,
  author = {J-P Martin and Lorenzo Alvisi and Michael Dahlin},
  title = {Small {Q}uorum {S}ystems},
  booktitle = {Proceedings of the 32nd International Conference on Dependable Systems and Networks (DSN'02)},
  pages = {374-383},
  year = {2002}
}

@inproceedings{eve,
    author = {Kapritsos, Manos and Wang, Yang and Quema, Vivien and Clement, Allen and Alvisi, Lorenzo and Dahlin, Mike},
    title = {All about {E}ve: execute-verify replication for multi-core servers},
    year = {2012},
    isbn = {9781931971966},
    publisher = {USENIX Association},
    address = {USA},
    booktitle = {Proceedings of the 10th USENIX Conference on Operating Systems Design and Implementation},
    pages = {237--250},
    numpages = {14},
    location = {Hollywood, CA, USA},
    series = {OSDI'12}
}

@article{smr,
    author = {Schneider, Fred B.},
    title = {Implementing fault-tolerant services using the state machine approach: a tutorial},
    year = {1990},
    issue_date = {Dec. 1990},
    publisher = {Association for Computing Machinery},
    address = {New York, NY, USA},
    volume = {22},
    number = {4},
    issn = {0360-0300},
    url = {https://doi.org/10.1145/98163.98167},
    doi = {10.1145/98163.98167},
    journal = {ACM Comput. Surv.},
    month = dec,
    pages = {299--319},
    numpages = {21}
}

@inproceedings{hotstuff,
    author = {Yin, Maofan and Malkhi, Dahlia and Reiter, Michael K. and Gueta, Guy Golan and Abraham, Ittai},
    title = {HotStuff: BFT Consensus with Linearity and Responsiveness},
    year = {2019},
    isbn = {9781450362177},
    publisher = {Association for Computing Machinery},
    address = {New York, NY, USA},
    url = {https://doi.org/10.1145/3293611.3331591},
    doi = {10.1145/3293611.3331591},
    booktitle = {Proceedings of the 2019 ACM Symposium on Principles of Distributed Computing},
    pages = {347--356},
    numpages = {10},
    location = {Toronto ON, Canada},
    series = {PODC '19}
}

@inproceedings{q/u,
    author = {Abd-El-Malek, Michael and Ganger, Gregory R. and Goodson, Garth R. and Reiter, Michael K. and Wylie, Jay J.},
    title = {Fault-scalable Byzantine fault-tolerant services},
    year = {2005},
    isbn = {1595930795},
    publisher = {Association for Computing Machinery},
    address = {New York, NY, USA},
    url = {https://doi.org/10.1145/1095810.1095817},
    doi = {10.1145/1095810.1095817},
    booktitle = {Proceedings of the Twentieth ACM Symposium on Operating Systems Principles},
    pages = {59--74},
    numpages = {16},
    location = {Brighton, United Kingdom},
    series = {SOSP '05}
}

@inproceedings {nysiad,
    author = {Chi Ho and Robbert Van Renesse and Mark Bickford and Danny Dolev},
    title = {Nysiad: Practical Protocol Transformation to Tolerate Byzantine Failures},
    booktitle = {5th USENIX Symposium on Networked Systems Design and Implementation (NSDI 08)},
    year = {2008},
    address = {San Francisco, CA},
    url = {https://www.usenix.org/conference/nsdi-08/nysiad-practical-protocol-transformation-tolerate-byzantine-failures},
    publisher = {USENIX Association},
    month = apr
}

@inproceedings{bftsmart,
    author = {Bessani, Alysson and Sousa, Jo\~{a}o and Alchieri, Eduardo E. P.},
    title = {State Machine Replication for the Masses with BFT-SMART},
    year = {2014},
    isbn = {9781479922338},
    publisher = {IEEE Computer Society},
    address = {USA},
    url = {https://doi.org/10.1109/DSN.2014.43},
    doi = {10.1109/DSN.2014.43},
    booktitle = {Proceedings of the 2014 44th Annual IEEE/IFIP International Conference on Dependable Systems and Networks},
    pages = {355--362},
    numpages = {8},
    series = {DSN '14}
}

@article{healthcare-blockchain,
    author = {Agbo, Cornelius C. and Mahmoud, Qusay H. and Eklund, J. Mikael},
    title = {Blockchain Technology in Healthcare: A Systematic Review},
    journal = {Healthcare},
    volume = {7},
    year = {2019},
    number = {2},
    article-number = {56},
    url = {https://www.mdpi.com/2227-9032/7/2/56},
    pubmed-id = {30987333},
    issn = {2227-9032},
    DOI = {10.3390/healthcare7020056}
}

@article{healthcare-blockchain2,
    author = {Angraal, Suveen and Krumholz, Harlan and Schulz, Wade},
    year = {2017},
    month = {09},
    pages = {e003800},
    title = {Blockchain Technology: Applications in Health Care},
    volume = {10},
    journal = {Circulation: Cardiovascular Quality and Outcomes},
    doi = {10.1161/CIRCOUTCOMES.117.003800}
}

@inproceedings{autobahn,
    author = {Giridharan, Neil and Suri-Payer, Florian and Abraham, Ittai and Alvisi, Lorenzo and Crooks, Natacha},
    title = {Autobahn: Seamless high speed BFT},
    year = {2024},
    isbn = {9798400712517},
    publisher = {Association for Computing Machinery},
    address = {New York, NY, USA},
    url = {https://doi.org/10.1145/3694715.3695942},
    doi = {10.1145/3694715.3695942},
    booktitle = {Proceedings of the ACM SIGOPS 30th Symposium on Operating Systems Principles},
    pages = {1--23},
    numpages = {23},
    location = {Austin, TX, USA},
    series = {SOSP '24}
}

@article{partial-sync,
    author = {Dwork, Cynthia and Lynch, Nancy and Stockmeyer, Larry},
    title = {Consensus in the presence of partial synchrony},
    year = {1988},
    issue_date = {April 1988},
    publisher = {Association for Computing Machinery},
    address = {New York, NY, USA},
    volume = {35},
    number = {2},
    issn = {0004-5411},
    url = {https://doi.org/10.1145/42282.42283},
    doi = {10.1145/42282.42283},
    journal = {J. ACM},
    month = apr,
    pages = {288--323},
    numpages = {36}
}

@misc{tpcc,
    author = {{Transaction Processing Performance Council}},
    howpublished = {\url{http://www.tpc.org/tpcc}},
    title = {The {TPC-C} home page}
}

@misc{protobuf,
    key = {Protocol Buffers},
    howpublished = {\url{https://protobuf.dev/}},
    title = {Protocol Buffers},
    note = {Accessed: 2025-11-04}
}

@misc{ed25519,
    author = {Daniel J.  Bernstein and Niels Duif and Tanja Lange and Peter Schwabe and Bo-Yin Yang},
    title = {High-speed high-security signatures},
    howpublished = {Cryptology {ePrint} Archive, Paper 2011/368},
    year = {2011},
    url = {https://eprint.iacr.org/2011/368}
}

@misc{ed25519-donna,
    author = {Andrew Moon},
    title = {ed25519-donna},
    howpublished = {\url{https://github.com/floodyberry/ed25519-donna}},
    note = {Accessed: 2025-11-04}
}

@misc{hmac,
    author = {Wei Dai},
    title = {CryptoPP},
    howpublished = {\url{https://github.com/weidai11/cryptopp}},
    note = {Accessed: 2025-11-04}
}

@misc{blake3,
    key = {blake3},
    title = {BLAKE3},
    howpublished = {\url{https://github.com/BLAKE3-team/BLAKE3}},
    note = {Accessed: 2025-11-04}
}

@misc{peloton,
    key = {peloton},
    title = {{Peloton-DB}},
    howpublished = {\url{https://db.cs.cmu.edu/peloton/}},
    note = {Accessed: 2025-11-04}
}

@misc{postgres,
    key = {{PostgreSQL}},
    title = {{PostgreSQL}},
    howpublished = {\url{https://www.postgresql.org/}}
}

@misc{ethereum,
    key = {Ethereum},
    howpublished = {\url{https://ethereum.org/}},
    title = {Ethereum},
    note = {Accessed: 2025-11-05}
}

@misc{bftsmart-github,
    key = {BFTSmart},
    howpublished = {\url{https://github.com/bft-smart/library}},
    title = {{Byzantine Fault-Tolerant (BFT) State Machine Replication (SMaRt) v1.2}},
    note = {Accessed: 2025-11-10},
}

@misc{diem,
    key = {Diem},
    title = {Diem},
    howpublished = {\url{https://en.wikipedia.org/wiki/Diem_(digital_currency)}},
    note = {Accessed: 2025-09-30},
}

@misc{digital-euro,
    key = {Digital Euro},
    title = {Digital Euro},
    howpublished = {\url{https://www.ecb.europa.eu/euro/digital_euro/html/index.en.html}},
    note = {Accessed: 2025-09-30},
}

@misc{fnality,
    key = {Fnality},
    title = {Fnality},
    howpublished = {\url{https://fnality.com/}},
    note = {Accessed: 2025-09-30},
}

@inproceedings{sbft,
    author={Golan Gueta, Guy and Abraham, Ittai and Grossman, Shelly and Malkhi, Dahlia and Pinkas, Benny and Reiter, Michael and Seredinschi, Dragos-Adrian and Tamir, Orr and Tomescu, Alin},
    booktitle={2019 49th Annual IEEE/IFIP International Conference on Dependable Systems and Networks (DSN)}, 
    title={SBFT: A Scalable and Decentralized Trust Infrastructure}, 
    year={2019},
    volume={},
    number={},
    pages={568--580},
    doi={10.1109/DSN.2019.00063}
}

@inproceedings{flexiblebft,
    author = {Malkhi, Dahlia and Nayak, Kartik and Ren, Ling},
    title = {Flexible Byzantine Fault Tolerance},
    year = {2019},
    isbn = {9781450367479},
    publisher = {Association for Computing Machinery},
    address = {New York, NY, USA},
    url = {https://doi.org/10.1145/3319535.3354225},
    doi = {10.1145/3319535.3354225},
    booktitle = {Proceedings of the 2019 ACM SIGSAC Conference on Computer and Communications Security},
    pages = {1041--1053},
    numpages = {13},
    location = {London, United Kingdom},
    series = {CCS '19}
}

@inproceedings{librabft,
    title={State Machine Replication in the Libra Blockchain},
    author={Mathieu Baudet and Avery Ching and Andrey Chursin and George Danezis and François Garillot and Zekun Li and Dahlia Malkhi and Oded Naor and Dmitri Perelman and Alberto Sonnino},
    year={2019},
    booktitle = {The Libra Association Technical Report}
}

@mastersthesis{tendermint,
    author = {Ethan Buchman, Jae Kwon, Zarko Milosevic},
    title = {Tendermint: Byzantine fault tolerance in the age of
blockchains},
    school = {The University of Guelph},
    year = 2016
}

@inproceedings{HQ,
    author = {Cowling, James and Myers, Daniel and Liskov, Barbara and  Rodrigues, Rodrigo and Shrira, Liuba},
    title = {{HQ} {R}eplication: {A} {H}ybrid {Q}uorum protocol for {B}yzantine fault tolerance},
    booktitle = {Proceedings of the 7th USENIX Symposium on Operating Systems Design and Implementation},
    series = {NSDI'06},
    year = 2006 
}

@inproceedings{aardvark,
    author = {Clement, Allen and Wong, Edmund and Alvisi, Lorenzo and Dahlin, Mike and Marchetti, Mirco},
    title = {Making {B}yzantine fault tolerant systems tolerate {B}yzantine faults},
    year = {2009},
    publisher = {USENIX Association},
    address = {USA},
    booktitle = {Proceedings of the 6th USENIX Symposium on Networked Systems Design and Implementation},
    pages = {153--168},
    numpages = {16},
    location = {Boston, Massachusetts},
    series = {NSDI'09}
}

@misc{solana,
    author = {Yakovenko, Anatoly},
    title = {Solana: A new architecture for a high performance blockchain v0.8.13},
    url = {https://solana.com/solana-whitepaper.pdf},
    year = {2018}
}

@misc{avalanche,
    author = {Sekniqi, Kevin and Laine, Daniel and Buttolph, Stephen and Sirer, Emin G\"{u}n},
    title = {Avalanche Platform},
    url = {https://cdn.prod.website-files.com/5d80307810123f5ffbb34d6e/6008d7bbf8b10d1eb01e7e16_Avalanche%20Platform%20Whitepaper.pdf},
    year = {2020}
}

@inproceedings{jflow,
    author = {Myers, Andrew C.},
    title = {JFlow: practical mostly-static information flow control},
    year = {1999},
    isbn = {1581130953},
    publisher = {Association for Computing Machinery},
    address = {New York, NY, USA},
    url = {https://doi.org/10.1145/292540.292561},
    doi = {10.1145/292540.292561},
    booktitle = {Proceedings of the 26th ACM SIGPLAN-SIGACT Symposium on Principles of Programming Languages},
    pages = {228--241},
    numpages = {14},
    location = {San Antonio, Texas, USA},
    series = {POPL '99}
}

@article{hrdb,
    author = {Vandiver, Ben and Balakrishnan, Hari and Liskov, Barbara and Madden, Sam},
    title = {Tolerating byzantine faults in transaction processing systems using commit barrier scheduling},
    year = {2007},
    issue_date = {December 2007},
    publisher = {Association for Computing Machinery},
    address = {New York, NY, USA},
    volume = {41},
    number = {6},
    issn = {0163-5980},
    url = {https://doi.org/10.1145/1323293.1294268},
    doi = {10.1145/1323293.1294268},
    journal = {SIGOPS Oper. Syst. Rev.},
    month = oct,
    pages = {59--72},
    numpages = {14}
}

@book{concurrency-control-recovery-db,
    author    = {Philip A. Bernstein and Vassos Hadzilacos and Nathan Goodman},
    title     = {Concurrency Control and Recovery in Database Systems},
    publisher = {Addison-Wesley},
    year      = {1987},
    isbn      = {0-201-10715-5}
}

@inproceedings{critique-ansi-sql-isolation,
    author    = {Hal Berenson and Philip A. Bernstein and Jim Gray and Jim Melton and Elizabeth O'Neil and Patrick O'Neil},
    title     = {A Critique of {ANSI} {SQL} Isolation Levels},
    booktitle = {Proceedings of the 1995 ACM SIGMOD International Conference on Management of Data},
    year      = {1995},
    pages     = {1--10},
    publisher = {ACM},
    address   = {San Jose, CA, USA}
}

@inproceedings{heterogeneous-paxos,
    author = {Isaac C. Sheff and Xinwen Wang and Robbert van Renesse and Andrew C. Myers},
    title = {Heterogeneous Paxos},
    booktitle = {24th International Conference on Principles of Distributed Systems, {OPODIS} 2020, Strasbourg, France (Virtual Conference), December 14-16, 2020},
    series = {LIPIcs},
    volume = {184},
    pages = {5:1--5:17},
    publisher = {Schloss Dagstuhl - Leibniz-Zentrum f{\"{u}}r Informatik},
    year = {2020},
    url = {https://doi.org/10.4230/LIPIcs.OPODIS.2020.5},
    doi = {10.4230/LIPICS.OPODIS.2020.5},
    bibsource = {dblp computer science bibliography, https://dblp.org}
}

@techreport{biba-integrity,
    author       = {Kenneth J. Biba},
    title        = {Integrity Considerations for Secure Computer Systems},
    institution  = {MITRE Corporation},
    type         = {Technical Report},
    number       = {MTR-3153},
    year         = {1977},
    address      = {Bedford, MA},
}

@inproceedings{noninterference,
    author    = {Joseph A. Goguen and Jos{\'e} Meseguer},
    title     = {Security Policies and Security Models},
    booktitle = {Proceedings of the 1982 IEEE Symposium on Security and Privacy},
    pages     = {11--20},
    year      = {1982},
    publisher = {IEEE Computer Society},
}

@article{language-based-ifc,
  author  = {Andrei Sabelfeld and Andrew C. Myers},
  title   = {Language-Based Information-Flow Security},
  journal = {IEEE Journal on Selected Areas in Communications},
  volume  = {21},
  number  = {1},
  pages   = {5--19},
  year    = {2003}
}

@inproceedings{ifdb,
    author = {Schultz, David and Liskov, Barbara},
    title = {IFDB: decentralized information flow control for databases},
    year = {2013},
    isbn = {9781450319942},
    publisher = {Association for Computing Machinery},
    address = {New York, NY, USA},
    url = {https://doi.org/10.1145/2465351.2465357},
    doi = {10.1145/2465351.2465357},
    booktitle = {Proceedings of the 8th ACM European Conference on Computer Systems},
    pages = {43--56},
    numpages = {14},
    location = {Prague, Czech Republic},
    series = {EuroSys '13}
}

@article{spanner,
    author = {Corbett, James C. and Dean, Jeffrey and Epstein, Michael and Fikes, Andrew and Frost, Christopher and Furman, J. J. and Ghemawat, Sanjay and Gubarev, Andrey and Heiser, Christopher and Hochschild, Peter and Hsieh, Wilson and Kanthak, Sebastian and Kogan, Eugene and Li, Hongyi and Lloyd, Alexander and Melnik, Sergey and Mwaura, David and Nagle, David and Quinlan, Sean and Rao, Rajesh and Rolig, Lindsay and Saito, Yasushi and Szymaniak, Michal and Taylor, Christopher and Wang, Ruth and Woodford, Dale},
    title = {Spanner: Google's Globally Distributed Database},
    year = {2013},
    issue_date = {August 2013},
    publisher = {Association for Computing Machinery},
    address = {New York, NY, USA},
    volume = {31},
    number = {3},
    issn = {0734-2071},
    url = {https://doi.org/10.1145/2491245},
    doi = {10.1145/2491245},
    journal = {ACM Trans. Comput. Syst.},
    month = aug,
    articleno = {8},
    numpages = {22}
}

@inproceedings{cockroachdb,
    author = {Taft, Rebecca and Sharif, Irfan and Matei, Andrei and VanBenschoten, Nathan and Lewis, Jordan and Grieger, Tobias and Niemi, Kai and Woods, Andy and Birzin, Anne and Poss, Raphael and Bardea, Paul and Ranade, Amruta and Darnell, Ben and Gruneir, Bram and Jaffray, Justin and Zhang, Lucy and Mattis, Peter},
    title = {CockroachDB: The Resilient Geo-Distributed SQL Database},
    year = {2020},
    isbn = {9781450367356},
    publisher = {Association for Computing Machinery},
    address = {New York, NY, USA},
    url = {https://doi.org/10.1145/3318464.3386134},
    doi = {10.1145/3318464.3386134},
    booktitle = {Proceedings of the 2020 ACM SIGMOD International Conference on Management of Data},
    pages = {1493--1509},
    numpages = {17},
    location = {Portland, OR, USA},
    series = {SIGMOD '20}
}

@inproceedings{omniledger,
    author={Kokoris-Kogias, Eleftherios and Jovanovic, Philipp and Gasser, Linus and Gailly, Nicolas and Syta, Ewa and Ford, Bryan},
    booktitle={2018 IEEE Symposium on Security and Privacy (SP)}, 
    title={OmniLedger: A Secure, Scale-Out, Decentralized Ledger via Sharding}, 
    year={2018},
    volume={},
    number={},
    pages={583-598},
    doi={10.1109/SP.2018.000-5}
}

@phdthesis{adya-thesis,
    title  = {Weak Consistency: A Generalized Theory and Optimistic Implementations for Distributed Transactions},
    author = {Adya, Atul},
    school = {Massachusetts Institute of Technology},
    year   = {1999},
}

@book{db-sys-textbook,
    title={Database Systems: The Complete Book},
    author={Garcia-Molina, Hector and Ullman, Jeffrey D. and Widom, Jennifer},
    year={2008},
    edition={2nd},
    publisher={Pearson},
    isbn={978-0131873254}
}

@misc{sintr-tr,
      title={Sintr: Safe Interactive Transactions in the Presence of Byzantine Clients}, 
      author={Austin T. Li and Daniel H. Lee and Lorenzo Alvisi and Natacha Crooks and Florian Suri-Payer},
      year={2026},
      eprint={2608.27091},
      archivePrefix={arXiv},
      primaryClass={cs.DC},
      url={https://arxiv.org/abs/2608.27091}, 
}

\clearpage
\appendix

\section{Proofs}

\subsection{Formal Definitions}

We begin by restating our definitions for an application with interactive transactions, our notion of safety for application executions, and noninterference.

\begin{restateDefinition}{def:application}
    \defApplication{}
\end{restateDefinition}

\begin{restateDefinition}{def:validity}
    \defValidity{}
\end{restateDefinition}

\begin{restateDefinition}{def:noninterference}
    \defNoninterference{}
\end{restateDefinition}

Note that with policy versions, noninterference must be clarified.
\sys{} ensures writes cannot circumvent the latest policy, and further assumes that policy changes are proactive in ensuring safety.
Thus, reading a value written under an old policy version should be considered safe.

We also formally define what faithful execution is.
Intuitively, a faithful execution produces reads and updates that would be expected from the transaction specification.

\begin{definition}
    Consider a transaction $T$ with inputs $x$, read set $\readset{T}$, and write set $\writeset{T}$.
    Here we treat $\readset{T}$ as an ordered series of read results from object to value (in the order of transaction execution), and $\writeset{T}$ as an unordered collection of object-value pairs.
    $T$ is a \emph{faithful execution} of $T_i \in \mathcal{T}$ if given $x$ and $\readset{T}$, executing $T_i$ produces $\writeset{T_i}$ that exactly matches $\writeset{T}$.
    At each point in $T_i$ where an intermediate read result to an object $O$ is required, the next unused read result in $\readset{T}$ is used if it is a read to $O$; otherwise, the transaction was not faithfully executed.
\end{definition}

\subsection{BFT Backend Assumptions}


We will prove \sys{}'s guarantees in general, without assuming a specific backend system.
To do so, we make a minimal set of assumptions on the BFT database; the systems we instantiate \sys{} on satisfy them.

First, we clarify the notion of a committed transaction.
In some BFT backend systems, even correct replicas may diverge in local state~\cite{basil,pesto}.
As such, we use the term \emph{globally committed} to distinguish transactions that are considered to be committed by the BFT database as a whole.

\begin{definition}
    Each database replica has a set of transactions that it locally considers to be committed.
    From this, the backend protocol defines the set of \emph{globally committed} transactions.

    In an SMR-based system, all correct replicas contain the same sequence of committed transactions; this exactly determines the globally committed transactions.
    In an inconsistent replication system ({\em e.g.}, Basil and Pesto), transactions are globally committed if there is a quorum of replicas which consider it as committed ({\em e.g.}, at least $3f+1$ of $5f+1$ replicas).
\end{definition}

\sys{} requires that the backend BFT database's fault assumptions hold ({\em i.e.}, at most $f$ faulty replicas).
From this, we require the following to be true of the BFT database.
    
\begin{assumption}\label{assumption:committed-correct-replicas}
    A globally committed transaction must have some correct replica which considers it committed.
\end{assumption}

Assumption~\ref{assumption:committed-correct-replicas} gives a weak safety guarantee on the backend protocol.
SMR-based systems satisfy this as all correct replicas consider the same transactions to be commited; Basil/Pesto satisfy this with their quorum requirements. 

Next, \sys{} requires serializability\footnote{
For BFT databases, serializability is formally defined as Byz-serializability~\cite{basil}.
Intuitively, Byz-isolation ensures that correct clients observe state that could have been produced by correct clients alone.
} for its governance transactions.
Because \sys{} performs its concurrency control for governance transactions at each replica, it requires that the BFT backend achieves serializability through its correct replicas.

\begin{assumption}\label{assumption:iso-correct-replicas}
    If correct replicas only commit transactions which are serializable, then the set of globally committed transactions are also serializable.
\end{assumption}

The SMR-based systems satisfy this as the set of globally committed transactions is exactly those which every correct replica commits.
Basil and Pesto each achieve this as well; see Lemma 1 and Theorem 1 in each.

The BFT backend system should also ensure that any two globally committed transactions have at least one correct replica which has locally commited both.
Intuitively, this is the only way in which correct replicas alone can determine the isolation level of the database.

\begin{assumption}\label{assumption:conflict-iso-correct-replicas}
    For any two globally committed transactions, there exists at least one correct replica which has locally committed both.
\end{assumption}

\subsection{Correctness Sketch}

We sketch the main theorems and lemmas necessary to show \sys{}'s safety properties.
We proceed in three steps: we first prove safety under a fixed application policy, then we show governance transactions are serializable, and finally we prove safety for versioned policies.

First, we prove that \sys{}'s protocol ensures correct validation clients only give matching endorsements for valid transactions.

\begin{restateLemma}{lemma:val-client-endorse-valid}
    \lemmaValClientEndorseValid{}
\end{restateLemma}

We then show that if \sys{}'s policy holds for an object, then the protocol ensures the object is only updated by valid committed transactions.

\begin{restateLemma}{lemma:object-safety-fixed-policy}
    \lemmaObjectSafetyFixedPolicy{}
\end{restateLemma}

We also show \sys{}'s protocol guarantees noninterference for objects that are not lifted.

\begin{restateLemma}{lemma:nonint-fixed-policy}
    \lemmaNonintFixedPolicy{}
\end{restateLemma}

Next, we prove that governance transactions are serializable.

\begin{theorem}\label{thm:gov-txn-serializable}
    \thmGovTxnSerializable{}
\end{theorem}

Finally, we show that \sys{}'s protocol still guarantees safety with versioned policies.

\begin{restateTheorem}{thm:sys-safety}
    \thmSysSafety{}
\end{restateTheorem}

For noninterference with versioned policies, we first prove that governance transactions maintain noninterference.

\begin{restateLemma}{lemma:nonint-gov-txn}
    \lemmaNonintGovTxn{}
\end{restateLemma}

From this, we obtain noninterference in general for \sys{}.

\begin{theorem}\label{thm:noninterference}
    \thmNoninterference{}
\end{theorem}

\subsection{Safety with Fixed Application Policy}

Assuming a fixed application policy, we prove that \sys{} is safe.
In particular, we show application execution validity and noninterference.

\begin{lemma}\label{lemma:val-client-endorse-valid}
    \lemmaValClientEndorseValid{}
\end{lemma}
\begin{proof}
    Suppose for the sake of contradiction that $T$ is not valid but it received a matching endorsement from a correct validation client $C_V$.
    Because $C_V$ is correct, it will only begin validation for transactions in $\mathcal{T}$.
    Thus, to receive an endorsement, the initiating client $C$ must have claimed $T$ to be some transaction $T_i \in \mathcal{T}$.

    We assume that $T_i$ is deterministic.
    Any local computation can only depend on the inputs to the transaction and intermediate database results.
    Since $C_V$ is correct, it will execute all local computation in $T_i$ according to its specification.
    Furthermore, because all forwarded database results contain \reqproof{} (which certify both \reqid{} and \reqres{}), $C_V$ will only use correct database results where the read object exactly matches what $T_i$ would require at that point.
    Thus, $C_V$ can only produce a faithful execution of $T_i$.
    However, since $T$ is not a faithful execution of $T_i$, it must have a write set that differs.
    The resulting endorsement hash contains the write set, so the endorsement $C_V$ gives cannot match transaction $T$.
    This is a contradiction.
\end{proof}

\begin{lemma}\label{lemma:object-safety-fixed-policy}
    \lemmaObjectSafetyFixedPolicy{}
\end{lemma}
\begin{proof}
    Suppose for the sake of contradiction that there exists an invalid committed transaction $T$ which updates $O$.
    By Assumption~\ref{assumption:committed-correct-replicas}, there exists at least one correct replica which committed $T$.
    At that replica, $T$ must have passed \sys{}'s endorsement check.
    Thus, $T$ contains at least $p_O$ matching endorsements 
    (Algorithm~\ref{alg:endorse-val}
    Lines~\ref{alg:endorse-val-line:for-write}--\ref{alg:endorse-val-line:endfor-write}).
    Because the replica checks for endorsement validity
    (Algorithm~\ref{alg:endorse-val}
    Lines~\ref{alg:endorse-val-line:for-endorse}--\ref{alg:endorse-val-line:endfor-endorse}), including the initiating client,
    there are at least $p_O+1$ clients which obtained the same result for $T$.
    The number of Byzantine clients is $f_O \le p_O$, so there must be at least one correct client which endorsed $T$, or the initiating client was correct.
    If the initiating client was correct, then $T$ must be valid. 
    Otherwise, by Lemma~\ref{lemma:val-client-endorse-valid}, since only valid transactions receive endorsements from correct clients, $T$ must be valid.
    In either case, this is a contradiction.
\end{proof}

\begin{lemma}\label{lemma:all-object-safety}
    \lemmaAllObjectSafety{}
\end{lemma}
\begin{proof}
    If all objects are only updated by valid committed transactions, then all non-read-only committed transactions are valid, which is exactly the definition of application execution validity.
\end{proof}



\begin{lemma}\label{lemma:nonint-fixed-policy}
    \lemmaNonintFixedPolicy{}
\end{lemma}
\begin{proof}
    Suppose for the sake of contradiction that \sys{} does not uphold noninterference in this case.
    Then there exists objects $O_1, O_2$ with policies $p_1, p_2$, respectively, such that $p_1 < p_2$ but $O_1$ influences $O_2$.
    In particular, this means there exists a committed transaction $T$ that reads from $O_1$ and writes to $O_2$.

    By Assumption~\ref{assumption:committed-correct-replicas}, there exists at least one correct replica which committed $T$.
    At that replica, $T$ must have passed the policy information flow check (Algorithm~\ref{alg:policy-leak}).
    Since $O_1$ is not lifted, it is included in the read set for the information flow check.
    When the check reaches $O_1$ in the read set, its policy $p_1$ will be extracted (Line~\ref{alg:policy-leak-line:get-read-policy}).
    But Line~\ref{alg:policy-leak-line:if-implies} will explicitly compare $p_1$ against the maximum policy of the write set, which is at least $p_2$.
    A correct replica must then abort $T$.
    This is a contradiction.
\end{proof}

\subsection{Governance Transaction Serializability}


\sys{} requires serializability for governance transactions.
To formally define how a governance transaction can conflict with a regular transaction, we extend Adya's formalism~\cite{adya-thesis}.

An execution produces a direct serialization graph (DSG) whose vertices are committed transactions, denoted $T_t$, where $t$ is the unique timestamp identifier. 
Edges between regular transactions in the DSG are one of three types:
\begin{itemize}
    \item $T_i \xrightarrow{ww} T_j$ if $T_i$ writes the version of object $x$ that precedes $T_j$ in the version order.
    \item $T_i \xrightarrow{wr} T_j$ if $T_i$ writes the version of object $x$ that precedes $T_j$ reads.
    \item $T_i \xrightarrow{rw} T_j$ if $T_i$ reads the version of object $x$ that precedes $T_j$'s write.
\end{itemize}
We assume, as does Adya, that if an edge exists between $T_i$ and $T_j$, then $T_i \ne T_j$. 

For an object $x$, let $p_x$ be its policy.
When a regular transaction writes to $x$, it uses (implicitly reads) a version of $p_x$ for its endorsement check.
The edges which can exist between a governance transaction $G$ and a regular transaction $T$ are:
\begin{itemize}
    \item $G_i \xrightarrow{wr} T_j$ if $G_i$ writes the version of $p_x$ that is implicitly read in $T_j$'s write to $x$.
    \item $T_i \xrightarrow{rw} G_j$ if $T_i$ writes a version of $x$ that implicitly reads the version of $p_x$ that precedes $G_j$'s write.
\end{itemize}



We assume that governance transactions are write-only and issued one at a time by the application.
As such, we exclude edges between governance transactions for the DSG.


We prove here that \sys{}'s governance transactions, when integrated into a serializable BFT database with optimistic concurrency control (OCC), are serializable.

\begin{restateTheorem}{thm:gov-txn-serializable}
    \thmGovTxnSerializable{}
\end{restateTheorem}
\begin{proof}
    We first show that at each correct replica, the set of committed transactions are serializable. 
    To do this, we show that if there exists an edge $T_i \to T_j$ in the DSG, then $i < j$. 
    Since \sys{} does not modify regular transaction concurrency control, we only need to show this holds between regular and governance transactions.

    \myhdr{Case: $G_i \xrightarrow{wr} T_j$}
    $T_j$ only implicitly reads policy versions with timestamp before $j$, and thus we must have $i < j$.

    \myhdr{Case: $T_i \xrightarrow{rw} G_j$}
    Suppose for the sake of contradiction there exists such an edge but $j < i$. 
    In particular, let $T_i$ be the regular transaction and let $G_j$ the governance transaction.
    We split our analysis based on whether the replica ran the CC check on $T_i$ first or $G_j$ first, and in each case show that this violation is not possible.

    \begin{enumerate}
        \item 
            \myhdr{CC check on $G_j$ first}
            During the CC check on $T_i$, each write set object policy is checked for conflicts against (replica-voted) committed governance transactions 
            (Algorithm~\ref{alg:gov-txn-cc} Lines~\ref{alg:gov-txn-cc-line:for-regular-vs-gov}--\ref{alg:gov-txn-cc-line:endfor-regular-vs-gov}).
            However, $G_j$ is already voted to be committed, and this check would have found $ts_p < j < i$ and aborted $T_i$.
            This is a contradiction.
        \item 
            \myhdr{CC check on $T_i$ first}
            When the CC check on $G_j$ occurs, $T_i$ is (replica-voted) committed and thus has recorded an implicit policy read to $p$ at time $i$ 
            (Algorithm~\ref{alg:gov-txn-cc} Lines~\ref{alg:gov-txn-cc-line:for-add-implicit}--\ref{alg:gov-txn-cc-line:endfor-add-implicit}).
            When the CC check for $G_j$ occurs, $G_j$ will not be allowed to commit since it would have caused $T_i$'s implicit policy read to miss a policy write 
            (Algorithm~\ref{alg:gov-txn-cc} Lines~\ref{alg:gov-txn-cc-line:for-check-implicit}--\ref{alg:gov-txn-cc-line:endfor-check-implicit}).
            This is a contradiction.
    \end{enumerate}

    In all cases, we have a contradiction.
    This implies $i < j$.

    Thus, in the DSG for a correct replica, if there exists an edge $T_i \to T_j$, then $i < j$.
    The set of transactions is serializable if the DSG has no cycles.
    Assume for the sake of contradiction that there exists a cycle consisting of $n$ transactions $T_{ts_1}, \dots, T_{ts_n}$. 
    This implies that $ts_1 < \dots < ts_n < ts_1$.
    However, transaction timestamps are totally ordered.
    This is a contradiction, and thus a correct replica only commits transactions which are serializable.

    Then by Assumption~\ref{assumption:iso-correct-replicas}, if every correct replica only commits transactions which are serializable, then the set of committed transactions is serializable.
\end{proof}

\subsection{Safety with Policy Versions}

We now prove that \sys{} maintains safety with policy versions.

\begin{lemma}\label{lemma:object-safety}
    \lemmaObjectSafety{}
\end{lemma}
\begin{proof}
    Suppose for the sake of contradiction that there exists an invalid committed transaction $T$ which updates $O$.
    Let $G$ be the latest governance transaction relative to $T$ which updated the policy $p_O$.
    By Assumption~\ref{assumption:conflict-iso-correct-replicas}, there exists at least one correct replica which committed $T$ and $G$.
    By Theorem~\ref{thm:gov-txn-serializable}, that replica must have ordered $G$ before $T$.
    Thus, $T$ must have passed \sys{}'s endorsement checks relative to $p_O$.
    The rest of the proof follows by the same logic as in Lemma~\ref{lemma:object-safety-fixed-policy}.
\end{proof}

\begin{restateTheorem}{thm:sys-safety}
    \thmSysSafety{}
\end{restateTheorem}
\begin{proof}
    By Lemmas~\ref{lemma:object-safety} and \ref{lemma:all-object-safety} (note that Lemma~\ref{lemma:all-object-safety} does not depend on a fixed policy).
\end{proof}

\begin{lemma}\label{lemma:replica-seen-two-gov-txns}
    \lemmaReplicaSeenTwoGovTxns{}
\end{lemma}
\begin{proof}
    We prove this for the systems we instantiate \sys{} on.

    For the SMR-based systems, all transactions are ordered through BFT consensus.
    Thus, every correct replica has committed the same set of transactions; in particular, the most recent governance transactions will be committed.

    For Basil/Pesto, we use a quorum argument.
    $G_1$ and $G_2$ must have been committed by all correct replicas, or at least $n-f = 4f+1$ replicas.
    Thus, there must be at least $3f+1$ replicas which have committed both $G_1$ and $G_2$.
    In order for $T$ to be globally committed, there must exist at least $3f+1$ replicas which committed it.
    Thus, there are at least $f+1$ replicas (at least one correct replica) which commit $T$ and both $G_1$ and $G_2$. 
\end{proof}

\begin{lemma}\label{lemma:nonint-gov-txn}
    \lemmaNonintGovTxn{}
\end{lemma}
\begin{proof}
    Suppose for the sake of contradiction that there exists a committed transaction which violates noninterference.
    Let $T$ be such a transaction with timestamp $ts_T$, and let $T$ read object $O_1$ with version $ts_{O_1}$ and write object $O_2$.
    In the global serialization order, let $O_1$ have policy $p_1$ with version $ts_{p_1}$ and let $O_2$ have policy $p_2$ with version $ts_{p_2}$.
    Let $G_1$ be the governance transaction which creates $p_1$'s version $ts_{p_1}$ and let $G_2$ be the governance transaction which creates $p_2$'s version $ts_{p_2}$.
    To violate noninterference, there must be a quorum of replicas which commits $T$ despite $p_1 < p_2$.

    By Lemma~\ref{lemma:replica-seen-two-gov-txns}, there exists at least one correct replica in this quorum which has committed both $G_1$ and $G_2$ and votes to commit $T$.
    We will first prove that this correct replica must have used $p_1$ for $O_1$ in its information flow check.
    We split our analysis based on whether the replica ran the CC check on $T$ first or $G_1$ first.
    Additionally, let $T'$ be the transaction that creates $O_1$'s version $ts_{O_1}$.


    \myhdr{Case 1: CC check on $G_1$ first}
    If $G_1$ is globally committed by the time the information flow check on $T$ runs, then the check will capture policy $p_1$ with version $ts_{p_1}$ for $O_1$ 
    (Algorithm~\ref{alg:policy-leak-versioned} Line~\ref{alg:policy-leak-versioned-line:get-read-policy}). 

    Now consider the case when $G_1$ is only locally committed by the time the information flow check on $T$ runs.
    If this were possible, then $T$ would not read $ts_{p_1}$ as the latest version for $p_1$.
    However, since $T$ depends on $T'$ for its read to $O_1$, $T'$ must have been globally prepared or committed by that point as well.
    By Theorem~\ref{thm:gov-txn-serializable}, we know that if $T'$ and $G_1$ both commit, then they are serializable.
    If they both globally commit, since $ts_{p_1} < ts_{O_1}$, $T'$ must be held to the latest policy in $G$.
    This means that $G_1$ must already be globally committed by the time $T'$ prepares (and, transitively, $T$ as well).
    Otherwise, one of $G_1$ or $T'$ would be aborted.
    This contradicts $G_1$ only being prepared.

    \myhdr{Case 2: CC check on $T$ first}
    Since $T$ depends on $T'$ for its read to $O_1$, $T'$ must be globally prepared or committed before $T$.
    But then $T'$ cannot be held to the latest policy in $G_1$.
    In order to maintain serializability (Theorem~\ref{thm:gov-txn-serializable}), $G_1$ must have aborted in this case.
    This is a contradiction.

    Thus, it must be the case at least one correct replica used $p_1$ for $O_1$ in its information flow check.
    By Theorem~\ref{thm:gov-txn-serializable}, if both $G_2$ and $T$ commit, then the replica must use $p_2$ for the write to object $O_2$.
    But we have $p_1 < p_2$, so the correct replica must abort the transaction
    (Algorithm~\ref{alg:policy-leak-versioned} Line~\ref{alg:policy-leak-versioned-line:if-implies}).
    This is a contradiction.
\end{proof}

\begin{restateTheorem}{thm:noninterference}
    \thmNoninterference{}
\end{restateTheorem}
\begin{proof}
    By Lemma~\ref{lemma:nonint-fixed-policy}, \sys{} does not violate noninterference for fixed policies.
    By Lemma~\ref{lemma:nonint-gov-txn}, governance transactions do not introduce violations for noninterference.
    Thus, in all cases, \sys{} does not violate noninterference.
\end{proof}

\section{Extended Technical Discussion}

This section outlines supplementary technical details, optimizations, and considerations for \sys{} beyond those discussed in the main technical sections.

\subsection{Transaction Commit with Versioned Policies}

We show \sys{}'s endorsement validation and policy information flow checks with a versioned policy store in Algorithms~\ref{alg:endorse-val-versioned} and \ref{alg:policy-leak-versioned}, respectively.
We highlight the differences below.

For endorsement validation, for each write set key, the replica gets the latest policy version relative to $T$ 
(Line~\ref{alg:endorse-val-versioned-line:get-policy}).
Assuming transaction $T$ has timestamp $ts_T$, the replica gets the latest policy earlier than $ts_T$.

For policy information flow, the maximum write set policy is similarly modified to use the latest policy for each key
(Line~\ref{alg:policy-leak-versioned-line:max-write-set-policy}).
For each read set key, the replica checks what policy was in place when the key was written 
(Line~\ref{alg:policy-leak-versioned-line:get-read-policy}).
If there exists a more recent version of that policy, then this implies that the version in the read set was admitted under an old policy 
(Lines~\ref{alg:policy-leak-versioned-line:if-latest-policy}--\ref{alg:policy-leak-versioned-line:endif-latest-policy}).
In this case we assume the read value was safe and do not check for information flow.

\begin{algorithm}[t]
\caption{Endorsement Validation ($T$)}\label{alg:endorse-val-versioned}
\begin{algorithmic}[1]
\State $E_T = \{\}$
\For{$e \in \sysEndorses{T}$}
    \If{$\mathit{e.sig} \text{ is valid} \land \mathit{e.hash} == H_T$}
        \State $E_T = E_T \cup e$
    \EndIf
\EndFor
\For{$\mathit{key} \in \writeset{T}$}
    \State $p = \mathit{GetPolicy(key, \ ts_T)}$\label{alg:endorse-val-versioned-line:get-policy}
    \If{$|E_T| < p$}
        \State \Return $\False$
    \EndIf
\EndFor
\State \Return $\True$
\end{algorithmic}
\end{algorithm}

\begin{algorithm}[t]
\caption{Policy Information Flow ($T$)}\label{alg:policy-leak-versioned}
\begin{algorithmic}[1]
\State $P_T = \max(\{\mathit{GetPolicy(key, \ ts_T)}: \ \mathit{key} \in \writeset{T}\})$\label{alg:policy-leak-versioned-line:max-write-set-policy}
\For{$\mathit{key, \ version} \in \readset{T}$}
    \State $p, \ \mathit{ts_p} = \mathit{GetPolicy(key, \ version)}$\label{alg:policy-leak-versioned-line:get-read-policy}
    \If{$\lnot$ ($\mathit{ts_p}$ is the latest version of $p$)}\label{alg:policy-leak-versioned-line:if-latest-policy}
        \State \textbf{continue}
    \EndIf\label{alg:policy-leak-versioned-line:endif-latest-policy}
    \If{$p < P_T$}\label{alg:policy-leak-versioned-line:if-implies}
        \State \Return $\False$
    \EndIf
\EndFor
\State \Return $\True$
\end{algorithmic}
\end{algorithm}

\subsection{Unique Transaction Identifier with Endorsements}

In \sys{}, endorsements give a client direct control over a replica's decision on its transaction.
For some BFT databases ({\em e.g.}, Basil~\cite{basil}, Pesto~\cite{pesto}), this gives Byzantine clients undue influence on correct client transaction outcomes.
We use Basil as a representative example to illustrate when the vulnerability is possible and how \sys{} addresses it.

Basil makes two design choices which, when na\"{i}vely combined with \sys{}'s endorsements, creates a liveness attack for Byzantine clients.
First, Basil makes clients the coordinators of their own transactions.
This prevents correct clients from relying upon Byzantine replicas, but also introduces the risk for Byzantine clients to equivocate transaction results.
However, Byzantine clients can only equivocate transaction results when replica responses allow them to do so; it is not a strategy that can be pursued reliably in Basil without \sys{}.
Second, for performance, transactions are permitted to read prepared (but not committed) transaction writes.
If correct client transactions acquire dependencies on equivocated transactions, they can subsequently be aborted.

To perform this attack, a Byzantine client can send their transaction with a sufficient set of endorsements to the minimum number of replicas to make it visible ($f+1$ in Basil).
After the transaction is made visible, the Byzantine client then sends the same transaction, but with insufficient endorsements, to a distinct set of replicas to abort the transaction (at least $3f+1$ in Basil).
Correct replicas, upon receiving this transaction, will vote to abort it due to the transaction failing the endorsement check. 
A correct client that acquired a dependency on this transaction while it was visible will eventually be forced to abort its own transaction.

To prevent this, \sys{} ensures that equivocated transactions are considered unique.
Upon commit, a transaction $T$ receives a unique identifier $id_T$, which is computed as a cryptographic hash of $T$ to prevent Byzantine clients from manipulating $T$'s content. 
\sys{} modifies the transaction identifier to also include endorsements. 
Identical transactions with different sets of endorsements will be considered different transactions by the system.
Aborting a transaction with a different set of endorsements will not affect the results of the original transaction.

\subsection{Sharing Validation Information}

In \sys{}, initiating clients must share information about their transaction to validation clients.
In particular, clients must share what transactions they run, and what reads and writes are made during the transaction.
For some BFT systems ({\em e.g.}, Basil~\cite{basil}, Pesto~\cite{pesto}), this introduces the possibility for Byzantine validation clients to attack the liveness of the system by creating conflicting transactions.
We use Basil as a representative example to discuss the vulnerability.

Basil allows clients to choose the timestamps of their transactions for concurrency control. 
This allows clients to choose when their transactions commit in the serializable schedule of transactions.
Byzantine clients can design read-only transactions with conflict windows that overlap with the original transactions.
Since \sys{} does not require read-only transactions to collect endorsements, the likelihood that the conflicting transaction commits before the original transaction is high.

Requiring endorsements for read-only transactions can prevent this attack.
Application execution validity is then extended to read-only transactions, preventing Byzantine clients from creating arbitrary read-only transactions.
However, committing read-only transactions becomes more expensive as they would now also incur validation costs.

We expect \sys{} to be used in settings where Byzantine attacks can occur but are infrequent, an assumption commonly held in permissioned blockchains.
In order to uphold their reputation in a permissioned system, Byzantine clients are unlikely to participate in detectable Byzantine attacks, and prefer to preserve liveness of the system if they cannot break safety~\cite{peerreview,flexiblebft}.
Under such assumptions, applications do not need to extend \sys{} with endorsements for read-only transactions.

\subsection{Generalizing Policies}

In principle, \sys{} is compatible with any integrity policy that describes whether a transaction with a set of endorsements can modify an object.
We formally define this generalization as follows.

\begin{definition}
    Let the set of clients be $[n] = \{1, 2, \dots, n\}$.
    A \emph{policy} is as a Boolean-valued function that takes a subset of clients as input, and outputs true if a transaction endorsed by that subset of clients is allowed to write, and false otherwise.
    That is, a policy $p$ is a function $p: 2^{[n]} \to \{0,1\}$. 
\end{definition}

For a set of client endorsements $E \subseteq [n]$, if $p(E) = 1$, then the set of endorsements satisfies the object's policy.
Otherwise, if $p(E) = 0$, it does not.
For example, a policy that requires at least two clients to endorse any transaction could be written as $p(E) = (|E| \ge 2)$. 

Next, we generalize information flow accordingly.
As before, a transaction that reads from an object of policy $p_1$ and writes to an object of policy $p_2$ has safe information flow if $p_1$ is at least as strong as $p_2$. 
We formally define this partial ordering below.

\begin{definition}    
    For two policies $p_1, p_2$, we write $p_1 \sqsubseteq p_2$ to mean that $p_1$ is \emph{at least as strong as} $p_2$.
    \begin{align*}
        p_1 \sqsubseteq p_2 \iff \forall S \subseteq [n], (p_1(S) = 1) \implies (p_2(S) = 1)
    \end{align*}
\end{definition}

More complex policies over client sets can cause endorsements to grow rapidly when accessing objects with non-overlapping policies. 
This creates practical challenges for performance and liveness because of the dependence on specific clients.

\end{document}